\documentclass[aps,11pt,twoside, nofootinbib]{revtex4}
\usepackage{amsthm}
\usepackage{amsmath,latexsym,amssymb,verbatim,enumerate,graphicx}
\usepackage{color}

\newtheorem{definition}{Definition} 
\newtheorem{prop}[definition]{Proposition}
\newtheorem{lemma}[definition]{Lemma}

\newtheorem{thm}[definition]{Theorem}
\newtheorem{corollary}[definition]{Corollary}

\newtheorem*{rep@theorem}{\rep@title}
\newcommand{\newreptheorem}[2]{%
\newenvironment{rep#1}[1]{%
 \def\rep@title{#2 \ref{##1} (restatement)}%
 \begin{rep@theorem}}%
 {\end{rep@theorem}}}
\makeatother

\newreptheorem{thm}{Theorem}
\newreptheorem{lem}{Lemma}
\newreptheorem{prop}{Proposition}

\def\ba#1\ea{\begin{align}#1\end{align}}
\def\ban#1\ean{\begin{align*}#1\end{align*}}

\newcommand{\ot}{\otimes}
\newcommand{\be}{\begin{equation}}
\newcommand{\ee}{\end{equation}}

\def\cor{\text{Cor}}

\def\benum{\begin{enumerate}}
\def\eenum{\end{enumerate}}

\def\squareforqed{\hbox{\rlap{$\sqcap$}$\sqcup$}}
\def\qed{\ifmmode\squareforqed\else{\unskip\nobreak\hfil
\penalty50\hskip1em\null\nobreak\hfil\squareforqed
\parfillskip=0pt\finalhyphendemerits=0\endgraf}\fi}
\def\endenv{\ifmmode\;\else{\unskip\nobreak\hfil
\penalty50\hskip1em\null\nobreak\hfil\;
\parfillskip=0pt\finalhyphendemerits=0\endgraf}\fi}

\newcommand{\bra}[1]{\langle #1|}
\newcommand{\ket}[1]{|#1\rangle}
\newcommand{\braket}[2]{\langle #1|#2\rangle}
\newcommand{\tr}{\text{tr}}

\newcommand{\id}{\mathbb{I}}

\newcommand{\<}{\langle}
\renewcommand{\>}{\rangle}

\def\id{{\operatorname{id}}}

\DeclareMathOperator{\rank}{rank}

\def\be{\begin{equation}}
\def\ee{\end{equation}}
\def\ben{\begin{eqnarray}}
\def\een{\end{eqnarray}}
\def\ot{\otimes}

\def\bei{\begin{itemize}}
\def\eei{\end{itemize}}

\def\E{{\mathbb{E}}}

\def\ep{\epsilon}

\mathchardef\ordinarycolon\mathcode`\:
\mathcode`\:=\string"8000
\def\vcentcolon{\mathrel{\mathop\ordinarycolon}}
\begingroup \catcode`\:=\active
  \lowercase{\endgroup
  \let :\vcentcolon
  }

\newcommand{\nc}{\newcommand}
 \nc{\proj}[1]{|#1\rangle\!\langle #1 |} 
\nc{\avg}[1]{\langle#1\rangle}

\nc{\conv}{\operatorname{conv}}
\nc{\smfrac}[2]{\mbox{$\frac{#1}{#2}$}} \nc{\Tr}{\operatorname{Tr}}
\nc{\ox}{\otimes} \nc{\dg}{\dagger} \nc{\dn}{\downarrow}
\nc{\lmax}{\lambda_{\text{max}}}
\nc{\lmin}{\lambda_{\text{min}}}

\nc{\csupp}{{\operatorname{csupp}}}
\nc{\qsupp}{{\operatorname{qsupp}}} \nc{\var}{\operatorname{var}}
\nc{\rar}{\rightarrow} \nc{\lrar}{\longrightarrow}
\nc{\poly}{\operatorname{poly}}
\nc{\polylog}{\operatorname{polylog}} \nc{\Lip}{\operatorname{Lip}}
\nc{\Om}{\Omega}
\nc{\wt}[1]{\widetilde{#1}}

\def\>{\rangle}
\def\<{\langle}

\nc{\glneq}{{\raisebox{0.6ex}{$>$}  \hspace*{-1.8ex} \raisebox{-0.6ex}{$<$}}}
\nc{\gleq}{{\raisebox{0.6ex}{$\geq$}\hspace*{-1.8ex} \raisebox{-0.6ex}{$\leq$}}}

\nc{\vholder}[1]{\rule{0pt}{#1}}
\nc{\wh}[1]{\widehat{#1}}
\nc{\h}[1]{\widehat{#1}}

\nc{\ob}[1]{#1}

\def\beq{\begin {equation}}
\def\eeq{\end {equation}}

\def\be{\begin{equation}}
\def\ee{\end{equation}}

\nc{\eq}[1]{(\ref{eq:#1})} 
\nc{\eqs}[2]{\eq{#1} and \eq{#2}}

\nc{\eqn}[1]{Eq.~(\ref{eqn:#1})}
\nc{\eqns}[2]{Eqs.~(\ref{eqn:#1}) and (\ref{eqn:#2})}

\nc{\region}{\cS\cW}

\begin{document}

\title{{\Large Exponential Decay of Correlations Implies Area Law }}

\author{Fernando G.S.L. Brand\~ao}
\email{fgslbrandao@gmail.com}
\affiliation{Institute for Theoretical Physics, ETH Z\"urich, 8093 Z\"urich, Switzerland}

\author{Micha\l{} Horodecki}
 \email{fizmh@ug.edu.pl}
\affiliation{Institute for Theoretical Physics and Astrophysics, University of Gda\'nsk, 80-952 Gda\'nsk, Poland}


\begin{abstract}

We prove that a finite correlation length, i.e. exponential decay of correlations, implies an area law for the entanglement entropy of quantum states defined on a line.
The entropy bound is exponential in the correlation length of the state, thus reproducing as a particular case Hastings proof of an area law for groundstates of 1D gapped Hamiltonians. 

As a consequence, we show that 1D quantum states with exponential decay of correlations have an efficient classical approximate description as a matrix product state of polynomial bond dimension, thus giving an equivalence between injective matrix product states and states with a finite correlation length. 


The result can be seen as a rigorous justification, in one dimension, of the intuition that states with exponential decay of correlations, usually associated with non-critical phases of matter, are simple to describe. It also has implications for quantum computing: It shows that unless a pure state quantum computation involves states with long-range correlations, decaying at most algebraically with the distance, it can be efficiently simulated classically. 

The proof relies on several previous tools from quantum information theory -- including entanglement distillation protocols achieving the hashing bound, properties of single-shot smooth entropies, and the quantum substate theorem -- and also on some newly developed ones. In particular we derive a new bound on correlations established by local random measurements, and we give a generalization to the max-entropy of a result of Hastings concerning the saturation of mutual information in multiparticle systems. The proof can also be interpreted as providing a limitation on the phenomenon of data hiding in quantum states.

\end{abstract}

\maketitle

\parskip .75ex


\section{Introduction}

Quantum states of many particles are fundamental to our understanding of many-body physics. Yet they are extremely daunting objects, requiring in the worst case an exponential number of parameters in the number of subsystems to be even approximately described. How then can multiparticle quantum states be useful for giving predictions to physical observables? The intuitive explanation, based on several decades of developments in condensed matter physics and more recently also on complementary input from quantum information theory, is that physically relevant quantum states, defined as the ones appearing in nature, are usually much simpler than generic quantum states. In this paper we prove a new theorem about quantum states that gives further justification to this intuition. 

One physical meaningful way of limiting the set of quantum states is to put restrictions on their \textit{correlations}. Given a bipartite quantum state $\rho_{XY}$, we can quantify the correlations between $X$ and $Y$ by 
\begin{equation}
\label{eq:cor}
\text{Cor}(X:Y) := \max_{\Vert M \Vert \leq 1, \Vert N \Vert \leq 1} |\tr((M \otimes N)(\rho_{XY} - \rho_{X} \otimes \rho_{Y}))|.
\end{equation}
where $\Vert M \Vert$ is the operator norm of $M$. Such a correlation function generalizes the more well-known two-point correlation function, widely studied in condensed matter physics, in which both $X$ and $Y$ are composed of a single site.

We say a quantum state $\rho_{1, ..., n}$ composed of $n$ qubits defined on a finite dimensional lattice has $(\xi, l_0)$-exponential decay of correlations if
for every $l \geq l_0$ and every two regions $X$ and $Y$ separated by more than $l$ sites (see Fig. \ref{fig00}),
\begin{equation}
\text{Cor}(X:Y) \leq 2^{- l/\xi}.
\end{equation} 
Here $\xi$ is the correlation length of the state and $l_0$ the minimum length for which correlations start decreasing. Such a form of exponential decay of correlations is sometimes also termed the exponential clustering property (see e.g. \cite{AHR62, Fre85, NS06}).

\begin{figure}  
\begin{center}  
\includegraphics[height=7cm,width=10cm,angle=0]{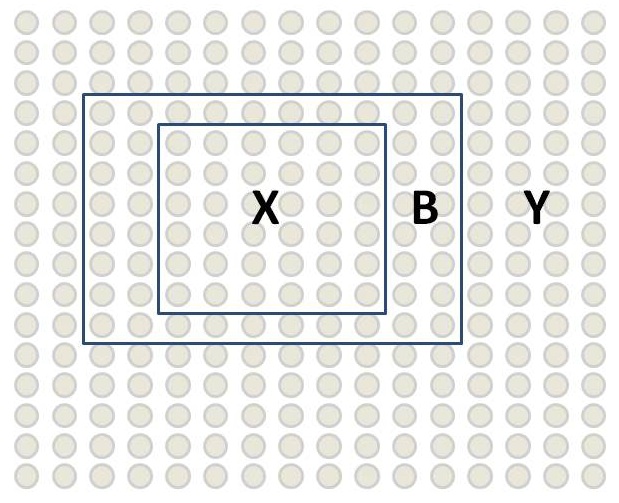}  
\caption{\small Example of a partition of the lattice into regions $X$, $B$ and $Y$. The separation between $X$ and $Y$ in the example consists of one site. \label{fig00}}  
\end{center}  
\end{figure}  

An important class of quantum states with exponential decay of correlations are groundstates of non-critical, gapped, local quantum Hamiltonians. In a seminal work, Hastings proved that in any fixed dimension the groundstate of a gapped local Hamiltonian exhibits exponential decay of correlations, with $l_0$ of  unit order and the correlation length $\xi$ of order of the inverse spectral gap of the model \cite{Has04} (see also \cite{Has04b, Has04c, HT06, NS06} for subsequent developments and \cite{AHR62, Fre85} for a previous analogous result in the context of relativistic quantum systems). Recently a simpler proof of exponential decay of correlations, of a combinatorial flavour, was given in Ref. \cite{AALV10} for the class of frustration-free models with a unique groundstate. 

Another physically motivated way of limiting the set of quantum states is to put restrictions on their \textit{entanglement} \cite{HHHH08}. Given a quantum state $\ket{\psi}$ we say it satisfies an area law for entanglement entropy if for every region $X$, $S(\rho_X) \leq O(\partial X )$, where $\rho_X$ is the reduced density matrix of $\ket{\psi}$ in the region $X$, $S(X) := - \tr(\rho_X \log \rho_X)$ the von Neumann entropy, and $\partial X$ the perimeter of the region $X$ \cite{ECP10}. The term entanglement entropy is due to the interpretation of $S(X)$ as the amount of entanglement between $X$ and $\overline{X}$ (the complementary region) in the quantum pure state $\ket{\psi}_{X\overline{X}}$ \cite{BBPS96}. Generic states typically have an extensive behaviour of entanglement entropy. Thus states satisfying an area law have a strong restriction on their amount of entanglement.

What kind of quantum states are expected to obey an area law? Starting with the seminal work of Bekenstein in the context of black hole entropy \cite{Bek73}, and more recently in the context of quantum spin systems \cite{VLRK03, CC04} and quantum harmonic systems \cite{AEPW02, PEDC05, Wol06} (see \cite{ECP10} for more references), an increasing body of evidence appeared suggesting that states corresponding to the ground or to low-lying energy eigenstates of local models satisfy an area law \footnote{With a logarithmic correction in the case of critical systems (see e.g. Ref. \cite{ECP10}).}$^{,}$\footnote{For thermal states of local Hamiltonians the situation is simpler: Jaynes' principle of maximum entropy implies a general area law for mutual information valid in arbitrary dimensions, with the mutual information of a region with the complementary region being proportional to the perimeter of the region and the inverse temperature \cite{WFHC08}.}. It turns out however that this is not always the case: One can construct local Hamiltonians, even in one dimension \cite{DGIK09} and with translational symmetry \cite{Ira10, GI09}, whose groundstate entanglement entropy follows a volume law. Yet, these are critical models, with the spectral gap shrinking to zero with the number of sites. Gapped models, on the other hand, are expected to satisfy an area law, although a general proof in arbitrary dimensions is still lacking.

In a ground-breaking work, Hastings proved that this is indeed the case for one-dimensional systems, i.e. 1D gapped Hamiltonians with a unique groundstate always obey an area law \cite{Has07}. Hastings' proof is based on locality estimates in quantum spin Hamiltonians known as Lieb-Robinson bounds \cite{LR72} and gives an exponential dependency between the entanglement entropy and the spectral gap of the model. Recently, Arad, Kitaev, Landau and Vazirani found a combinatorial proof of the area law for groundstates of 1D gapped models, with a bound on entanglement entropy only polynomially large in the spectral gap \cite{AKLV12, ALV11}, matching (up to polynomial factors) explicit examples \cite{Ira10, GH10}. 

Both proofs of the area law explore extensively the fact that the quantum state is a groundstate of a local Hamiltonian with a constant spectral gap, which gives much more structure than merely the fact that the state has exponential decay of correlations. However, as pointed out already in Ref. \cite{VC06}, exponential decay of correlations by itself already suggests the entanglement of the state should satisfy an area law. Indeed, consider a quantum state $\ket{\psi}_{XBY}$ as in Fig. \ref{fig00}, with $B$ the boundary region between $X$ and $Y$. If $\ket{\psi}$ has exponential decay of correlations and the separation between $X$ and $Y$ is of order of the correlation length of the state, $X$ will have almost no correlations with $Y$, and one would expect the entanglement of $X$ with $BY$ to be only due to correlations with the region $B$, thus obeying an area law.

Perhaps surprisingly, and as also pointed out in Ref. \cite{VC06}, the argument presented above is flawed. This is because there are quantum states for which $X$ and $Y$ have almost no correlations, yet $X$ has very large entropy. This is not only a pathological case among quantum states, but it is actually the general rule: The overwhelming majority of quantum states will have such peculiar type of correlations, whenever the regions $X$ and $Y$ have approximately equal sizes \cite{HLW06}. Quantum states with this property are termed \textit{quantum data hiding} states, due to their use in hiding correlations (and information) from local measurements \cite{DLT02}. Developing a classification of entanglement in such states is one of the outstanding challenges in our understanding of quantum correlations \footnote{As a concrete example, in the problem of determining if a quantum bipartite state $\rho_{XY}$ is entangled or not, quantum data hiding states appear to be the hardest instances from a computational point of view \cite{BCY10, BCY10b}.}. 

We therefore see that it is not possible to obtain an area law from exponential decay of correlations by considering only one fixed partition of the system (into $XBY$ regions). However this leaves open the possibility that an area law could be established by exploring exponential decay of correlations simultaneously in several \textit{different} partitions of the system.

\section{Results}

\noindent \textbf{Area Law for Pure States:} Quantum data hiding states, and the related quantum expander states \cite{Has07b, Has07c}, have been largely recognized as an obstruction for obtaining an area law for entanglement entropy from exponential decay of correlations (see e.g. \cite{VC06, Has07, Has07b, Has07c, WFHC08, Mas09, AALV10, ALV11, Osb11}). In this paper we show that such an implication, at least for states defined on a 1-dimensional lattice, is in fact \textit{false}. Our main result is the following (see Fig. \ref{fig000}):

\def\entropybound{

Let $\ket{\psi}_{1, ..., n}$ be a state defined on a ring with $(\xi, l_0)$-exponential decay of correlations and $n \geq  C l_0 /\xi$. Then for any connected region $X \subset [n]$ and every $l\geq 8 \xi$,
\begin{equation} \label{entropybound}
H^{2^{- \frac{l}{8\xi}}} _{\max}(X) \leq c' l_{0}\exp\left( c \log(\xi)\xi \right) + l,
\end{equation}
with $C, c, c' > 0$ universal constants.}

\begin{thm} \label{maintheorempure}
\entropybound
\end{thm}

\begin{figure}  
\begin{center}  
\includegraphics[height=5cm,width=5cm,angle=0]{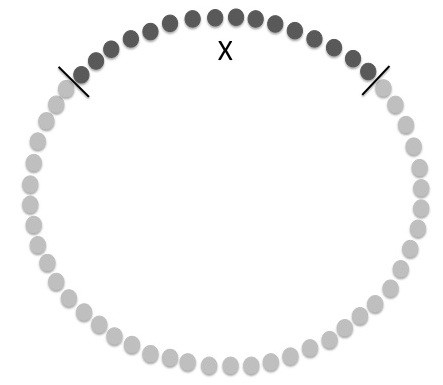}  
\caption{\small Arrangement of qubits on a ring considered in Theorem \ref{maintheorempure}. \label{fig000}}  
\end{center}  
\end{figure}  

An announcement of the proof can be found in \cite{BH13}. 

We leave it as an open question whether a similar statement holds true in dimensions larger than one.

The quantity $H_{\max}^{\varepsilon}(X)$ is the $\varepsilon$-smooth max-entropy of $\rho_X := \tr_{\backslash X}(\ket{\psi}\bra{\psi})$, where $\tr_{\backslash X}(\ket{\psi}\bra{\psi})$ denotes the partial trace of $\ket{\psi}\bra{\psi}$ with respect to all sites expect the ones in region $X$. The $\varepsilon$-smooth max-entropy is defined as \cite{TCR10, Ren05}
\begin{equation} 
H_{\max}^{\varepsilon}(X) =  H_{\max}^{\varepsilon}(\rho_X) := \min_{ \tilde \rho_X \in {\cal B}_{\varepsilon}(\rho_X)} \log \left( \text{rank}(\tilde \rho_X)  \right),
\end{equation}
with $B_{\varepsilon}(\rho_X) := \{ \tilde \rho_X : D(\rho_X, \tilde \rho_X) \leq \varepsilon \}$ the $\varepsilon$-ball of quantum states around $\rho_X$ and $D(\sigma_1, \sigma_2)$ the purified distance between states $\sigma_1$ and $\sigma_2$ \footnote{See Appendix \ref{purifieddistance} for the definition of the purified distance.}. In words, $H_{\max}^{\varepsilon}(X)$ gives the logarithm of the approximate support of $\rho_X$, i.e. the number of qubits needed to store an $\varepsilon$-approximation of the state $\rho_X$.

We now present two simple corollaries of Theorem \ref{maintheorempure}. The first is an area law for the von Neumann entropy:
\begin{corollary} \label{vonNeumannarealaw}
Let $\ket{\psi}_{1, ..., n}$ be a state defined on a ring with $(\xi, l_0)$-exponential decay of correlations and $n \geq  C l_0 /\xi$. Then for any connected region $X \subset [n]$,
\begin{equation} \label{vnbound}
S(X) \leq c'l_{0}\exp\left( c \log(\xi)\xi \right),
\end{equation}
with $C, c, c' > 0$ universal constants.
\end{corollary}
Note that by combining Corollary \ref{vonNeumannarealaw} with the fact that groundstates of gapped local Hamiltonians have exponential decay of correlations \cite{Has04}, we recover Hastings' area law for groundstates of gapped models, with the same exponential dependence of the entropy on the inverse spectral gap of the model \cite{Has07}.

The second corollary gives an approximation of a state $\ket{\psi}_{1, ..., n}$ in terms of a matrix product state of small bond dimension. A matrix product representation of the state $\ket{\psi}_{1,...,n}$ is given by 
\begin{equation}
\ket{\psi}_{1,...,n} = \sum_{i_1=1}^d ... \sum_{i_n=1}^d \tr(A^{[1]}_{i_1}...A^{[n]}_{i_n}) \ket{i_1, ..., i_n},
\end{equation}
with $D \times D$ matrices $A^{[j]}$, $j \in [n]$. The parameter $D$ is termed the bond dimension and measures the complexity of the matrix product representation. When $D = \poly(n)$ the quantum state $\ket{\psi}_{1,...,n}$ admits an efficient classical description in terms of its matrix product representation, with only polynomially many parameters and in which expectation values of local observables can be calculated efficiently. We call such states themselves matrix product states (MPS) \cite{FNW92, OR95, Vid03, PGVWC07}. 

It is known that every so-called injective matrix product state (which is the generic case) has a finite correlation length. Combining Theorem \ref{maintheorempure} and the results of Ref. \cite{Vid03} (see also \cite{VC06}) we find that the converse statement also holds true:

\begin{corollary} \label{MPS}
Let $\ket{\psi}_{1, ..., n}$ be a state defined on a ring with $(\xi, l_0)$-exponential decay of correlations. For every $\delta > 0$, there is a matrix product state $\ket{\phi_k}$ of bond dimension $k = \poly(n^{\xi \log(\xi)}, \delta)$ such that $|\braket{\psi}{\phi_k}| \geq 1 - \delta$.
\end{corollary}

Thus we see that one-dimensional pure quantum states with exponential decay of correlations have a very simple structure, admitting an efficient classical parametrization. In fact, such a representation is of more than theoretical interest, as matrix product states are the variational ansatz for the most successful known method for simulating one-dimensional quantum systems, the Density Matrix Renormalization Group (DMRG) \cite{Whi92, OR95}.

\vspace{0.2 cm}

\noindent \textbf{Area Law for Mixed States:} So far we have focused on the case of pure states satisfying exponential decay of correlations. How about mixed states? It is clear that Theorem \ref{maintheorempure} cannot be true in this case. As an example, consider the maximally mixed state: It has no correlations, yet the local entropies are all maximum. Although we do not present a full answer to the mixed state case in this paper, we can prove the following extension of Theorem \ref{maintheorempure}:

\def\entropyboundmixed{

Let $\rho_{1, ... , n}$ be a state defined on ring with $(\xi, l_0)$-exponential decay of correlations and $n \geq  C l_0 /\xi$. Then for any connected region $X \subset [n]$ and every $l\geq 8 \xi$,
\begin{equation} \label{entropybound}
H^{2^{- \frac{l}{8\xi}}} _{\max}(X) \leq c' l_{0}\exp\left( c \log(\xi)\xi \right)(1+H_{\max}(\rho)) + l,
\end{equation}
with $C, c, c' > 0$ universal constants.}
\begin{thm} \label{maintheoremmixed}
\entropyboundmixed
\end{thm}


\vspace{0.2 cm}

\noindent \textbf{Long-Range Correlations in Quantum Computation:} A perennial question in quantum information science is to understand what is responsible for the apparent superiority of quantum computation over classical computation. A fruitful approach in this direction is to find conditions under which quantum circuits have an efficient classical simulation (see e.g. \cite{Got98, JL02, Val02, dVT02, Vid03, MS08, AL08, Nest11}). In this way one can at least say what properties a quantum circuit \textit{should} have if it is supposed to give a superpolynomial speed-up over classical computing. 

In \cite{Vid03} Vidal gave an interesting result in this context: Unless a quantum computation in the circuit model involves states that violate an area law, with the entropy of a certain subregion being bigger than the logarithm of the number of boundary qubits, it can be simulated classically in polynomial time. A direct corollary of this result and Theorem \ref{maintheorempure} is the following:

\begin{corollary}
Consider a family of quantum circuits $V = V_{k}...V_2V_1$ acting on $n$ qubits arranged in a ring and composed of two qubit gates $V_k$. Let $\ket{\psi_t} := V_{t}...V_2V_1 \ket{0^n}$ be the state after the $t$-th gate has been applied. Then if there are constants $\xi, l_0$ independent of $n$ such that, for all $n$ and $t \in [n]$, $\ket{\psi_t}$ has $(\xi, l_0)$-exponential decay of correlations, one can classically simulate the evolution of the quantum circuit in $\poly(n, k)$ time.
\end{corollary}

The corollary says that one must have at least algebraically decaying correlations in a quantum circuit if it is supposed to solve a classically hard problem. Interestingly such long range correlations are usually associated to critical phases of matter. From a quantum information perspective, the result gives a limitation to the possibility of hiding information in 1D quantum circuits. 

\vspace{0.2 cm}

\noindent \textbf{Random States and Quantum Expanders:} It is interesting to analyse how Theorem \ref{maintheorempure} fits together with the entanglement properties of random states and quantum expander states, since these have been thought of as giving an obstruction to a statement of a similar flavour \cite{VC06, Has07, Has07b, Has07c, WFHC08, Mas09, AALV10, ALV11, Osb11}. 
\begin{figure}  
\begin{center}  
\includegraphics[height=4.5cm,width=6cm,angle=0]{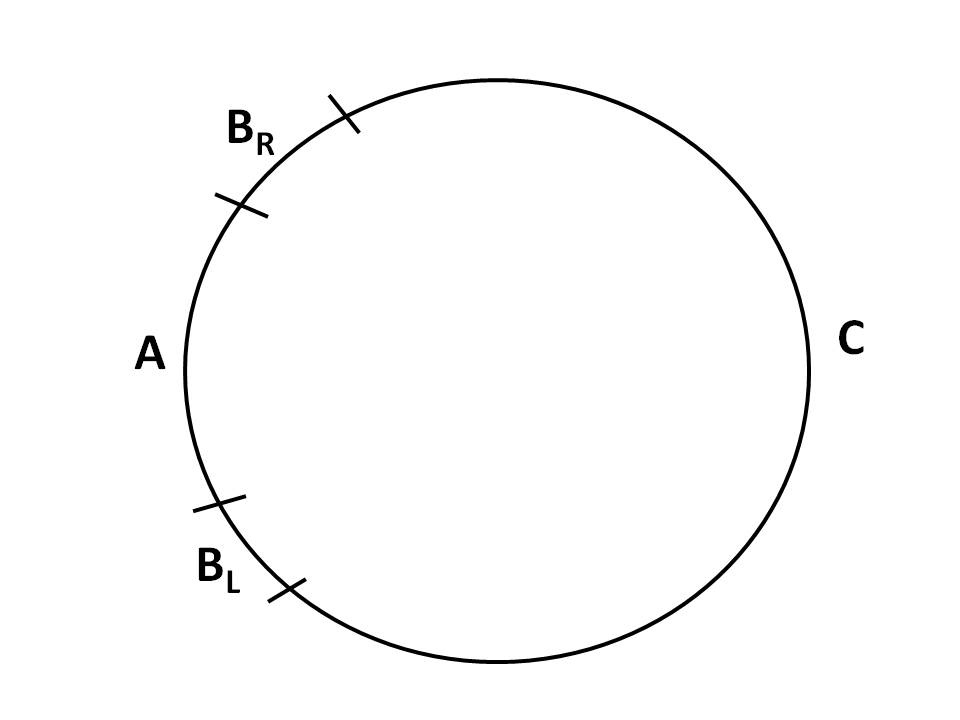}  
\caption{\small Partition of the one-dimensinal lattice into regions $A$, $B_L$, $B_R$, and $C$. \label{fig0}}  
\end{center}  
\end{figure}

\textit{Random States:} The situation for random quantum states is fairly simple: Even though they have exponentially small two-point correlations, they do \textit{not} exhibit general exponential decay of correlations, and thus nothing prevents the extensive behaviour of entanglement entropy found in such states. Indeed we show in Appendix  \ref{decouplingHaar} that for a quantum state $\ket{\psi}_{1,...,n}$ drawn from the Haar measure, with overwhelming probability, for every region $X$ with $|X| \leq n/4$, 
\begin{equation} \label{decouplingrandom}
D(\rho_X, \tau_X) \leq 2^{-\Omega(n)},
\end{equation}
where $\tau_X$ is the maximally mixed state on $X$. Now let us consider the sites $\{1, ..., n \}$ arranged in a ring and divide the region $X$ into two subregions $A$ and $B$ of equal size $n/8$, and call $C$ the region composed by the remaining $3n/4$ sites. The region $B$ is further divided into two subregions $B_L$ and $B_R$ of equal sizes, one to the left and the other to the right of $A$ (see Fig. \ref{fig0}). In quantum information terminology, Eq. (\ref{decouplingrandom}) says, that $A$ is approximately \textit{decoupled} from $B$, since their joint state is close to a product state \cite{ADHW09}. From Uhlmann's theorem (Lemma \ref{purifieddistancepurification} in Appendix \ref{purifieddistance}) it thus follows that there is an isometry $V : C \rightarrow C_1C_2$ that can be applied to $C$ such that
\begin{equation} \label{uhlmandecoupling0}
D( (\id_{AB} \otimes V) \ket{\psi}_{ABC}\bra{\psi}(\id_{AB} \otimes V)^{\cal y}, \ket{\Phi}_{AC_1}\bra{\Phi} \otimes \ket{\Phi}_{BC_2}\bra{\Phi}) \leq 2^{-\Omega(n)},
\end{equation}
with $\ket{\Phi}_{AC_1} = \text{dim(A)}^{-1/2} \sum_{k=1}^{\text{dim}(A)} \ket{k, k}$ a maximally entangled state between $A$ and $C_1$. Thus it is clear that the regions $A$ and $C$, separated by $n/16$ sites, are correlated, since $A$ and $C$ share a maximally entangled state (and are separated by a distance given by the size of $B_L$ and $B_R$). Indeed, choosing $M = \sum_{i=1}^{\text{dim}(A)/2} \ket{k}\bra{k}$ and $N = V \left( M \otimes \id_{C_2} \right) V^{\cal y}$, we find from Eq. (\ref{uhlmandecoupling0}) that 
\begin{equation}
\text{Cor}(A:C) \geq \tr \left((M_A \otimes N_C)(\rho_{AC} - \rho_A \otimes \rho_C)\right) = 1/4 - 2^{-\Omega(n)}. 
\end{equation}

\textit{Expander States:} Quantum expander states were introduced in Refs. \cite{Has07b, Has07c} in order to find a more explicit example of states in which exponential decay of correlations does not directly imply an area law. Let $\Lambda: {\cal D}(\mathbb{C}^D) \rightarrow {\cal D}(\mathbb{C}^D)$, with ${\cal D}(\mathbb{C}^D)$ the set of density matrices acting on $\mathbb{C}^D$, be a quantum channel with Kraus decomposition $\Lambda(\rho) = \sum_{i=1}^{d} A_i\rho A_{i}^{\cal y}$
for $\rho\in {\cal D}(\mathbb{C}^D)$. We say it is an $(\eta, d)$-quantum expander if its maximum eigenvalue (seen $\Lambda$ as a linear map) is one and its second largest eigenvalue is smaller than $\eta < 1$ \cite{Has07b, Has07c, BATS10} (we mean here eigenvalues of channel treated as a linear operator).  Given a quantum expander, its associated quantum expander state is defined as the translational-invariant matrix product state whose correlation matrices are given by the $\{ A_i \}_{i=1}^d$:
\begin{equation}
\ket{\psi}_{1, ..., n} := \sum_{i_1=1}^d ... \sum_{i_n=1}^d \tr(A_{i_1} ... A_{i_n}) \ket{i_1, ..., i_n}.
\end{equation}

Following the original approach of Fannes, Nachtergaele and Werner \cite{FNW92} (see also \cite{WFHC08}) we show in Appendix \ref{DecayMPS} that for regions $A$ and $C$ separated by $l$ sites:
\begin{equation} \label{smallcorrexpan}
\text{Cor}(A:C) \leq D \eta^{l}, 
\end{equation}
Thus the state $\ket{\psi}_{1,...,n}$ has at least $(O(\log^{-1}(1/\eta)), O(\log(D)\log^{-1}(1/\eta)))$-exponential decay of correlations.

For some choices of the matrices $\{ A_i \}_{i=1}^d$, one can show that the previous bound on the correlation length is tight. For that we use the approach of Ref. \cite{Has07b}, where it was shown that choosing $A_i = U_i/\sqrt{D}$, with $U_i$ drawn independently from the Haar measure, one obtains with high probability a $(d^{-1/2}, d)$-quantum expander. By a similar approach we show in Appendix \ref{CorrelationsExpander} that, with high probability, the corresponding quantum expander state is such that there \textit{are} strong correlations between regions separated by less than $c \log(D)/\log(d)$ sites, for a constant $c$. 

Therefore we have the following situation: There are correlations between regions separated by less than $c\log(D)/ \log(d)$ sites, for a constant c, while there is exponential decay of correlations
￼
for regions separated by more than $C \log(D)/ \log(d)$ sites, for another constant $C > c$. This is consistent with Theorem 1 and in fact shows that the linear dependence of $l_0$ in the entropy bound is optimal (since the entropy of quantum expander states saturates at approximately $\log(D)$). 

We note that one can also derive a bound on correlations in a MPS independent of the bond dimension, but only for regions which are sufficiently far away from the boundary of the chain. See \cite{Has14} for a derivation of the bound and a discussion of its relation with the results of this paper.

\section{Proof of Theorem \ref{maintheorempure}}

Before we turn to the proof in earnest, we first provide a sketch of it and an outline of the main techniques used, in order to give a more global view of its structure. In a nutshell, we explore the idea that if an area law is violated, then a \textit{random measurement} in one region, together with another carefully chosen measurement in another region far apart, give rise to stronger correlations than what would be allowed by the finite correlation length of the state. 

\vspace{0.2 cm}

\noindent \textbf{Proof Sketch:}  The key idea of the proof comes from the analysis provided before for random states. There we could find a partition of the system into three regions $ABC$ such that $A$ was approximately decoupled from $B$, and the state on $A$ was close to the maximally mixed state. Then by Uhlmann's theorem we could show that $A$ was strongly correlated with $C$. Therefore the fact that the entropy on $AB$ was close to maximum implied a lower bound on the correlations between $A$ and $C$. In the general case $AB$ might not have entropy close to maximum, and so $A$ might not be decoupled from $B$. However one can still try to follow the same reasoning by applying a measurement on $A$ that \textit{decouples} it from $B$.

The above problem -- of decoupling $A$ from $B$ while making the state on $A$ maximally mixed -- has been studied before and is known as entaglement distillation protocol \cite{DW03, HOW05, HOW07}. Given the state $\ket{\psi}_{ABC}^{\otimes n} \in ({\cal H}_A \otimes {\cal H}_B \otimes {\cal H}_C)^{\otimes n}$, with $n$ sufficiently large, consider a random measurement on $\left( {\cal H}_A \right)^{\otimes n}$, consisting of $N \approx 2^{n I(A:B)}$ Haar distributed projectors $\{ P_k \}_{k=1}^N$ of equal dimension summing up to the identity. Here $I(A:B) := H(A) + H(B) - H(AB)$ is the mutual information of $A$ and $B$. Then it was shown in Refs. \cite{DW03, HOW05, HOW07} that with high probability the post-selected state $\ket{\phi}_{A'B^nC^n} := (P_k \otimes \id_{B^nC^n})\ket{\psi}_{ABC}^{\otimes n}/\Vert  (P_k \otimes \id_{B^nC^n})\ket{\psi}_{ABC}^{\otimes n} \Vert$ is such that, if $H(B) \leq H(C)$,
\be
D \left( \rho_{A'B^n}, \tau_{A'} \otimes \rho_{B^n} \right) \approx 0 \hspace{0.2 cm} \text{and}  \hspace{0.2 cm}  |A'|\approx 2^{-n H(A|C)}, 
\ee
with $\tau_{A'}$ the maximally mixed state on $A'$ and $\rho_{A'B^n}$ the $A'B^n$ reduced density matrix of $\ket{\phi}_{A'B^nC^n}$. We say 
\be
- H(A|C) := H(C) - H(AC) = H(C) - H(B) 
\ee
is the entanglement distillation rate of the protocol, as it gives the number of EPR pairs shared by $A$ and $C$ after the random measurement on $A$. Here $H(A|C)$ is the conditional entropy of $A$ given the side information $C$. Using that the state $\ket{\psi}_{ABC}$ is pure the entanglement rate could alternatively be written as $H(A|B) = H(AB) - H(B)$ \footnote{A variant of this latter expression will be the one used later in the proof once we consider the single-shot entanglement distillation protocol.}. This bound on the entanglement distillation rate is known as hashing bound. 

Thus considering many copies of the state and making an appropriate measurement on $A$ we end up again in the situation where $A'$ is close to maximally mixed and decoupled from $B$, implying that if $H(A|C) < 0$, $A'$ is maximally entangled with (part of) $C$. The argument thus suggests that in order not to have long-range correlations between $A$ and $C$ one must have $H(C) \leq H(B)$, which gives an area law for region $C$ if $B$ has constant size.

There are two challenges for making this idea work. The first concerns the fact that the entanglement distillation protocol of \cite{DW03, HOW05, HOW07} is devised only in the limit of infinitely many copies of the state, but in our problem we have only a single copy of it. The second is the fact that we only get a particular outcome $k$ with probability $\approx 2^{- n I(A:B)}$, which could be catastrophic if we want to show that the correlations between $A$ and $C$ are large. Let us ignore the first challenge in this sketch and focus on the second. Thus we are going to assume that the entanglement distillation protocol works in the same way for a single copy of the state as it does asymptotically; later we will see how we can make this precise by using tools from single-shot quantum information theory \cite{Ren05, Tom12}. 

Considering the simplifying assumption that the distillation protocol works for a single copy of the state, the upshot of the result of \cite{DW03, HOW05, HOW07} is that $H(C) \geq H(B)$ implies $\text{Cor}(A:C) \geq 2^{-I(A:B)}$. Indeed, one could first make a random measurement on $A$, obtaining one of the possible outcomes with probability $2^{-I(A:B)}$ and distilling a maximally entangled state between $A$ and $C$, and then measure the correlations in the maximally entangled state.  We can also write the previous relation as
\begin{equation} \label{corfromentropybymerging}
\text{Cor}(A:C) \leq 2^{-I(A:B)}  \hspace{0.2 cm} \text{implies} \hspace{0.2 cm} H(C) \leq H(B).
\end{equation}

From exponential decay of correlations we have that $\text{Cor}(A:C) \leq 2^{-l / \xi}$, with $l = \text{size}(B)$. Thus if $I(A:B) \leq l /\xi$, we must have $H(C) \leq H(B)$, which constitutes an area law for region $C$ if $B$ has fixed size. Unfortunately we do not have any guarantee that $I(A:B) \leq l /\xi$. 

The second idea of the proof is to show that at a distance at most $l_0 \exp(O(\xi))$ sites from $C$, we can find a region $B$ of size smaller than $l_0 \exp(O(\xi))$ such that $H(B) \leq l/(2\xi)$, which implies $I(A:B) \leq 2H(B) \leq l /\xi$. Then by the argument above we can get a bound of  $l_0 \exp(O(\xi))$ on the entropy of a region which differs from $C$ by less than $l_0 \exp(O(\xi))$ sites, and so by subadditivity of entropy and the Araki-Lieb inequality \cite{AL70} we have a bound of $l_0 \exp(O(\xi))$ on the entropy of $C$ as well.

In order to prove the existence of a region $B$ with $H(B) \leq l/(2\xi)$ we apply a variant of a result due to Hastings \cite{Has07}, used in his proof of an area law for groundstates of 1D gapped models, concerning the saturation of mutual information in a multiparticle state. It appears as Lemma \ref{saturationmutualvonNeumann} in Appendix \ref{entropies} and states: For all $\varepsilon > 0$ and a particular site $s$ there exist neighbouring regions $X_{L}X_C X_R$ at a distance at most $l_0 \exp(O(1/\varepsilon))$ sites from $s$, with $X_L$ and $X_R$ each of size $l$ and $X_C$ of size $2l$, such that $I(X_C : X_LX_R) \leq \varepsilon l$ and $l \leq l_0 \exp(O(1/\varepsilon))$ (see Fig. \ref{sat}). Let us consider the partition $X_LX_CX_R R$, with $R$ the region composed of all the remaining sites not in $X_LX_CX_R$. Then we choose $\varepsilon = 1/(2 \xi)$ and use once more the fact that the state has exponential decay of correlations to find that $\text{Cor}(X_C:R) \leq 2^{-l / \xi} \leq 2^{I(X_C : X_LX_R)}$. Then using Eq. (\ref{corfromentropybymerging}) we get $H(R) \leq H(X_LX_R)$. But since $H(X_C) + H(X_LX_R) - H(R) = I(X_C : X_LX_R) \leq l/(2\xi)$, we find $H(X_C) \leq 1/(2\xi)$, which gives the desired relation setting $B = X_C$.

To summarize, we employ the entanglement distillation protocol in the form of Eq. (\ref{corfromentropybymerging}) and the assumption of exponential decay of correlations \textit{twice}. One in conjunction with the result about saturation of mutual information in order to get a region of constant size and not so large entropy, and the second to boost this into an area law for regions of arbitrary size. This finishes the sketch of the proof. $\square$

\begin{figure}  
\begin{center}  
\includegraphics[height=2cm,width=10.5cm,angle=0]{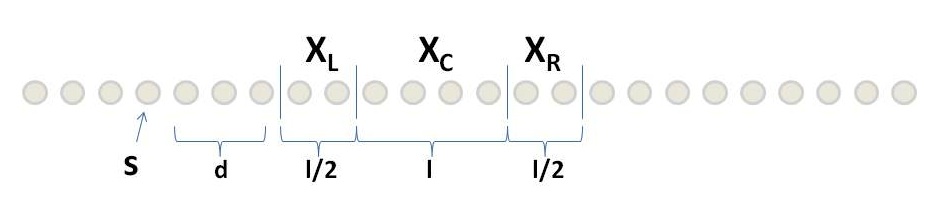}  
\caption{\small Structure of the regions for Lemma \ref{saturationmutualvonNeumann} and Lemma \ref{saturationmaxmutualinfo}. In Lemma \ref{saturationmutualvonNeumann}, $d, l \leq l_0\exp(O(1/\varepsilon))$, while in Lemma \ref{saturationmaxmutualinfo}, $d, l \leq l_0\exp(O(\log(1/\varepsilon)/\varepsilon))$. \label{sat}}  
\end{center}  
\end{figure}

In order to turn the sketch above into a proof, we will need two lemmas, which we believe might be of independent interest. Both provide single-shot analogues of the entropic results used in the sketch. To this goal we use the recent framework of single-shot quantum information theory \cite{Ren05, Tom12}, whose objective is to analyse the protocols of quantum information theory \textit{without} the assumption of having an arbitrarily large number of copies of the state (or uses of the channel in dynamical problems).

The first lemma gives three relations between correlations and smooth entropies in a general tripartite quantum pure state. To state it we will need the following single-shot analogue of the conditional entropy, which has an important role in the proof: Given a bipartite state $\rho_{AB} \in {\cal D}({\cal H}_A \otimes {\cal H}_B)$, the min-entropy of $A$ conditioned on $B$ is defined as
\begin{equation}
H_{\min}(A|B)_{\rho} := \max_{\sigma \in {\cal D}({\cal H}_B)} \sup \{ \lambda \in \mathbb{R} : 2^{-\lambda}\id_A \otimes \sigma_B \geq \rho_{AB}  \}.
\end{equation}
The $\varepsilon$-smooth min-entropy of $A$ conditioned on $B$ of $\rho_{AB}$ is given by
\begin{equation}
H_{\min}^{\varepsilon}(A|B)_{\rho} := \max_{\tilde \rho_{AB} \in {\cal B}_{\varepsilon}(\rho_{AB})} H_{\min}(A|B)_{\tilde \rho}.
\end{equation}


An operational interpretation for the conditional min-entropy, which we explore in the proof of the next lemma, was given in Ref. \cite{DBWR10} (see also Lemma \ref{singleshotmerging} in Appendix \ref{singleshotmergingsection} for a formal statement). Let $\ket{\psi}_{ABC}$ be an arbitrary purification of $\rho_{AC}$.

 $H_{\min}^{\varepsilon}(A|B)$ is the entanglement distillation rate of the single-shot protocol, where one has a \textit{single} copy of the state and allows an $\varepsilon$-error in the output state. Moreover, the classical communication cost of the protocol is given by the max-mutual information of $A$ with the purification $B$ of $AC$ \cite{Datta-Hsieh-imax} defined as \footnote{We note that there are other possible definitions of single-shot mutual information such as the one given in Ref. \cite{BCR11}.}
\begin{equation}  
I_{\max}^{\varepsilon}(A:B) := H_{\max}^{\varepsilon}(A) - H_{\min}^{\varepsilon}(A|B).
\end{equation}
\def\bigentropycor{
Given a tripartite state $\ket{\psi}_{ABC}$ and real numbers $\delta, \nu  >0$, the following holds
\begin{enumerate}
\item  $H_{\min}^{\delta}(A|B) \geq - 2 \log (\delta) + 5$ implies that
\begin{equation} \label{corfrommerging}
\text{{\rm Cor}}\left( A:C \right) \geq \left(\frac{1}{256} - 26 \sqrt{\delta} \right) 2^{- \left( I_{\max}^{\delta}(A:B) - 2 \log (\delta) + 5\right)}.
\end{equation}
\item $H_{\max}^{\nu}(C) \geq 3 H_{\max}^{\delta}(B)$ and $\frac{\nu}{2} - 2\delta \geq \sqrt{\frac{4\delta}{1-2\nu-2\delta}}$ imply that
\begin{equation}
\text{{\rm Cor}}\left( A:C \right) \geq \left( \frac\nu2-2\delta - \sqrt{\frac{4\delta}{1-2\nu-2\delta}} \right)^2 2^{-3\left( H^\delta_{\max}(B) + \log \left(\delta^{-2} + \nu^{-2} \right) + 5 \right)}.
\end{equation}
\item Setting $\gamma := \text{{\rm Cor}}(A:C)^{1/2}(1/2 - \delta)^{-1/2}$,
\be
H_{\max}^{2\gamma}(A) \leq H_{\max}^{\delta}(A) + 2 \log|B| + \log \left( \frac{2}{\gamma^2}  \right).
\label{eq:boost}
\ee
\end{enumerate}

}

\begin{lemma}[Correlations Versus Entropies]
\label{bigentropiesimplycorrelations}
\bigentropycor
\end{lemma}

Part 1 of Lemma \ref{bigentropiesimplycorrelations} is a direct consequence of the entanglement distillation protocol of Ref. \cite{DBWR10} (see Appendix \ref{singleshotmergingsection}). Part 3 also follows directly from the definition of correlation function together with basic properties of the max-entropy and the purified distance. 

Part 2 of the lemma, in turn, is a novel relation: Since one might have $H_{\max}^{\nu}(C) \geq 3 H_{\max}^{\delta}(B)$ but $H_{\min}^{\delta}(A|B) \leq 0$, it shows that even in situations where one cannot establish EPR pairs between two parties thus decoupling the third purifying party (which can only happen when $H_{\min}^{\delta}(A|B) > 0$ \cite{DBWR10}), one can still have large correlations between the two parties. Moreover, as shown in the proof such correlations are manifested by choosing a random measurement on $A$ and optimizing over a measurement on $C$. So we are operating in a regime where a random measurement on one party does not decouple it from either of the other two parties, but it does generate correlations with one of them.

The next lemma is an analogue for the max-mutual information of Lemma \ref{saturationmutualvonNeumann} in Appendix \ref{entropies}, which itself is a slight modification of a result originally shown by Hastings in his proof of an area law for groundstates of 1D gapped models \cite{Has07} 
\def\saturationImax{

Let $\rho_{1,...,n} \in {\cal D}\left( (\mathbb{C}^2)^{\otimes n} \right)$ be a state defined on a line with $(\xi, l_0)$-exponential decay of correlations, and let $s$ be a particular site. Then for all $\delta > 0$, $\min\{ \delta/2, 1/\xi\} \geq \varepsilon > 0$, and $\overline{l}_0 \geq l_0$, there is an $l$ satisfying $O(\xi^2 \log \left( 2/\varepsilon\right) \varepsilon^{-1}) \leq l/\overline{l}_0 \leq \exp(O(\log(1/\varepsilon)/\varepsilon))$ and a connected region $X_{2l} := X_{L, l/2} X_{C, l} X_{R, l/2}$ of $2l$ sites (the borders $X_{L, l/2}$ and $X_{R, l/2}$ with $l/2$ sites each, and the central region $X_{C, l}$ with $l$ sites) centred at most $\overline{l}_0\exp(O(\log(1/\varepsilon)/\varepsilon))$ sites away from $s$ (see Fig. \ref{sat}) such that
\begin{equation} \label{sublinearmutual}
I_{\max}^{\delta}(X_{C, l}:X_{L, l/2}X_{R, l/2}) \leq \frac{8 \varepsilon l}{\delta} + \delta.
\end{equation}

}

\begin{lemma} [Saturation of max-Mutual Information]
\label{saturationmaxmutualinfo}
\saturationImax
\end{lemma}

Although the proof of Lemma \ref{saturationmutualvonNeumann} in Appendix \ref{entropies} is straightforward, only involving a chain of applications of subadditivity of entropy, the proof of the single-shot counterpart given by Lemma \ref{saturationmaxmutualinfo} is considerably more involved. The main difficulty is that one might have big gaps between min- and max-entropies, and thus it is not clear how one can concatenate several uses of the subadditivity inequality. 

This difficulty is handled by exploring two ideas: The first is to apply the quantum substate theorem of Jain, Radhakrishnan, and Sen \cite{JJS07} (see also \cite{JN11} and Lemma \ref{thmsubstate} in Appendix \ref{entropies}) in order to turn the max-mutual information into a hybrid mutual information, involving both von Neumann and smooth max entropies. The second idea, necessary to handle the remaining gaps that might exist between von Neumann and max entropies, is to use the hypothesis of exponential decay of correlations, together with the quantum equipartition property \cite{TCR09, Tom12}, to upper bound the max-entropy by a sum of von Neumann entropies.


We are now ready to prove Theorem \ref{maintheorempure}. We will make use of the three parts of Lemma \ref{bigentropiesimplycorrelations}, each in conjunction with exponential decay of correlations, as well as Lemma \ref{saturationmaxmutualinfo}, which itself is based on exponential decay of correlations, which in total amounts to  \textit{four} applications of the assumption that the state has exponential decay of correlations.

\begin{repthm}{maintheorempure}
\entropybound
\end{repthm}

\begin{figure}  
\begin{center}  
\includegraphics[height=7cm,width=10cm,angle=0]{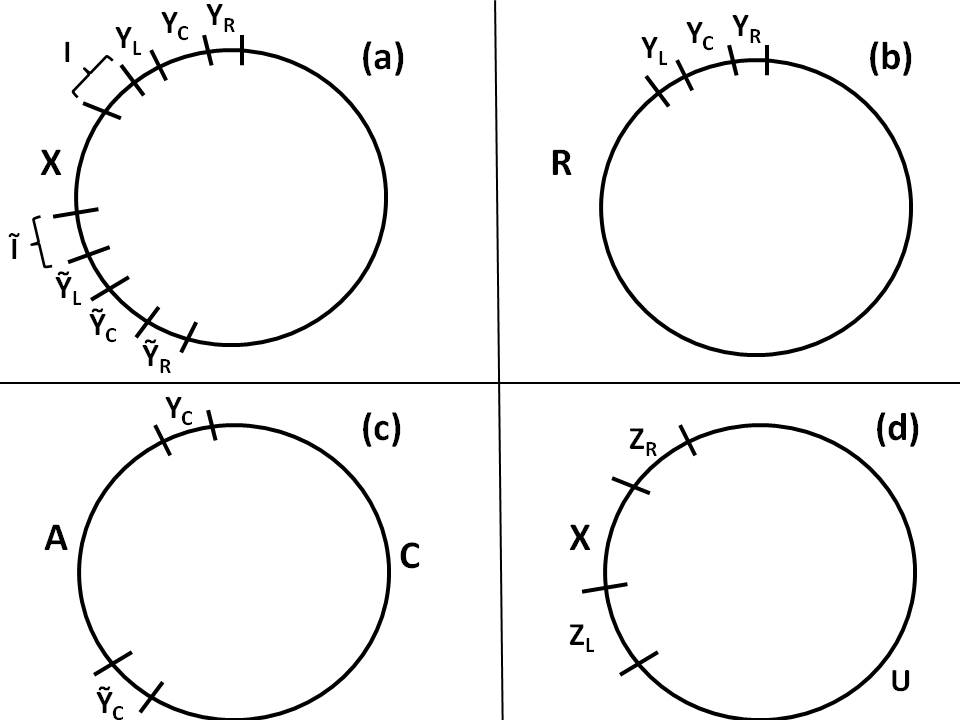}  
\caption{\small (a) Region $X$ separated from $Y := Y_LY_CY_R$ by $I \leq l_0\exp(O(\log(1/\varepsilon)/\varepsilon))$ sites and from $\tilde Y := \tilde Y_L \tilde Y_C \tilde Y_R$ by $\tilde I \leq  l_0\exp(O(\log(1/\varepsilon)/\varepsilon))$ sites; (b) partition into regions $Y := Y_LY_CY_R$ and $R$; (c) partition into regions $A$, $B := Y_C \tilde Y_C$, and $C$; (d) partition into regions $X$, $Z := Z_LZ_R$, and $U$. \label{fig2}}  
\end{center}  
\end{figure}  

\begin{proof}

We start by applying Lemma \ref{saturationmaxmutualinfo} twice to the two boundaries of the region $X$, with $\delta := 10^{-8}$, $\varepsilon := \min \left(1/(10^5 \xi), \delta/8 \right)$ and $\overline{l}_0 := \max \left( l_0, c l_0  \varepsilon/(\xi \log(2/\varepsilon)) \right)$, for a constant $c > 0$. Then we find that there are regions $Y$ and $\tilde Y$ (see Fig. \ref{fig2} (a)), each of size $2l$ with $O(\xi^2 \log \left( 2/\varepsilon\right))  \leq l/\overline{l}_0 \leq \exp(O(\log(1/\varepsilon)/\varepsilon))$, and at a distance at most $\overline{l}_0\exp(O(\log(1/\varepsilon)/\varepsilon))$ sites away from the two boundaries of $X$, respectively, such 
that 
\begin{equation}  \label{saturationmutualfirstpart11}
I^{\delta/4}_{\max}(Y_{C, l} : Y_{L, l/2}  Y_{R, l/2}) = H^{\delta/4}_{\max}(Y_{C, l}) - H^{\delta/4}_{\min}(Y_{C, l}| Y_{L, l/2}  Y_{R, l/2}) \leq \frac{l}{8 \xi},
\end{equation}
and likewise for $\tilde Y$.  We note that from the choice of parameters it follows by assuming that $l \geq 300 \xi$.

We now argue that exponential decay of correlations implies that the entropy of $\rho_{Y_{C, l}}$ must be small. Let $R$ be the complementary region to $Y$ (see Fig. \ref{fig2} (b)). Then from exponential decay of correlations we have $\text{Cor}(Y_{C, l} : R) \leq 2^{- l / (2\xi)}$. Applying part (1) of Lemma \ref{bigentropiesimplycorrelations} we find that
\begin{equation} \label{Hmaxeqbound1}
- H_{\min}^{\delta/4}(Y_{C, l}| Y_{L, l} Y_{R, l}) < -2 \log(\delta) + 9, 
\end{equation}
since otherwise
\begin{equation}
\text{Cor}(Y_{C, l} : R)  \geq \left(\frac{1}{256} - 26 \sqrt{\delta}\right) 2^{- I^{\delta/4}_{\max}(Y_{C, l} : Y_{L, l/2}  Y_{R, l/2}) - 2 \log(\delta) + 9 } \geq 2^{- \left( \frac{l}{8 \xi} - 34 \right)},
\end{equation}
in contradiction with the correlation length being $\xi$ (since $l \geq 300 \xi$). Then from Eq. (\ref{Hmaxeqbound1}) and Eq. (\ref{saturationmutualfirstpart11}) we have
\begin{equation}  \label{boundYc11}
H^{\delta/4}_{\max}(Y_{C, l}) \leq \frac{l}{8 \xi} -2 \log(\delta) + 9.
\end{equation}
Applying the same reasoning to $\tilde Y$ we find that also $H^{\delta/4}_{\max}(\tilde Y_{C, l}) \leq \frac{l}{8 \xi} - 2 \log(\delta) + 9$. 

Define $A$ as the region between $Y_{C, l}$ and $\tilde Y_{C, l}$, $B := Y_{C, l} \tilde Y_{C, l}$, and $C$ the complementary region to $AB$ (see Fig. \ref{fig2} (c)). Note that $A$ differs from $X$ by at most $2 \overline{l}_0\exp(O(\log(1/\varepsilon)/\varepsilon))$ sites and so by Lemma \ref{subadditivityHmax}, for every $\nu > 0$,
\begin{equation} \label{xregionboundbyA11}
H_{\max}^{2\nu}(X) \leq H_{\max}^{\nu}(A) + 2\overline{l}_0\exp(O(\log(1/\varepsilon)/\varepsilon)) + \log \frac{2}{\nu^2}.
\end{equation}

We now prove an upper bound on $H_{\max}^{\nu}(A)$, with $\nu = 0.01$. Since $A$ and $C$ are separated by $l$ sites, we have from exponential decay of correlations that 
\begin{equation} \label{corrlargesecondpart11}
\text{Cor}(A:C) \leq 2^{-l/\xi}. 
\end{equation}
From part (2) of Lemma \ref{bigentropiesimplycorrelations} and Eq. (\ref{corrlargesecondpart11}) we then find that
\begin{equation} \label{BonC11}
H_{\max}^{\nu}(C) < 3 H_{\max}^{\delta}(B),
\end{equation}
since otherwise, using that $l \geq 300\xi$,
\begin{equation}
\text{Cor}\left( A:C \right) \geq \left( \frac\nu2-2\delta - \sqrt{\frac{4\delta}{1-2\nu-2\delta}} \right)^2 2^{-3\left( H^\delta_{\max}(B) + 16 \log \left(\delta^{-2} + \nu^{-2} \right) + 5 \right)} \geq  2^{-3l/(4 \xi)},
\end{equation}
contradicting Eq. (\ref{corrlargesecondpart11}). In the last inequality of the equation above we used that
\begin{eqnarray}
H_{\max}^{\delta}(B) &=& H_{\max}^{\delta}(Y_{C, l} \tilde Y_{C, l}) \nonumber \\ &\leq& H_{\max}^{\delta/4}(Y_{C, l}) + H_{\max}^{\delta/4}(\tilde Y_{C, l}) + \log \frac{8}{\delta^2} \nonumber \\ &\leq& \frac{l}{4 \xi} - 4 \log(\delta) + 18  + \log \frac{8}{\delta^2},
\end{eqnarray}
where the first inequality follows from subadditivity of the max-entropy (Lemma \ref{subadditivityHmax}) and the second from Eq. (\ref{boundYc11}).

Using subadditivity of the max-entropy (Lemma \ref{subadditivityHmax}) again, we get
\begin{eqnarray}
H_{\max}^{3\nu + \delta}(A) = H_{\max}^{3\nu + \delta}(BC) &\leq& H_{\max}^{\delta}(B) +  H_{\max}^{\nu}(C) + \log \frac{2}{\nu^2} \nonumber \\ &\leq& 4H_{\max}^{\delta}(B) + \log \frac{2}{\nu^2} \nonumber \\ &\leq&  2\overline{l}_0\exp(O(\log(1/\varepsilon)/\varepsilon)),
\end{eqnarray}
where we used Eq. (\ref{BonC11}) and the bound $H_{\max}(B) \leq 2l \leq \overline{l}_0\exp(O(\log(1/\varepsilon)/\varepsilon))$. Therefore from Eq. (\ref{xregionboundbyA11}),
\begin{equation}
H_{\max}^{6\nu + 2\delta}(X) \leq 2\overline{l}_0\exp(O(\log(1/\varepsilon)/\varepsilon)) = O(l_0)\exp(O(\log(1/\varepsilon)/\varepsilon)).
\end{equation}

This is already an area law for $X$, although with a fixed error (recall that we fixed $\delta=10^{-8}$ and $\nu=0.01$). To finish the proof we show how applying exponential decay of correlations once more we can reduce the error. Let $Z$ be a region of size $2l$ separating $X$ from the remaining sites by a distance $l$, and denote by $U$ the complementary region to $XZ$ (see Fig. \ref{fig2} (d)) \footnote{Note that the distance $l$ has no relation with the parameter $l$ used in the first part of the proof and in particular it does not have to be bigger than $300 \xi$.}. By exponential decay of correlations we have $\text{Cor}(X:U) \leq 2^{-l/\xi}$. Using part (3) of Lemma \ref{bigentropiesimplycorrelations} we then get
\begin{equation}
H_{\max}^{2^{-l/(4\xi)}}(X)  \leq H_{\max}^{2^{-l/(2\xi) + 1}}(X) \leq H_{\max}^{6\nu + 2\delta}(X) + 2l  \leq O(l_0) \exp(O(\log(1/\varepsilon)/\varepsilon)) + 2l,
\end{equation}
where the first inequality follows assuming $l \geq 4\xi$. 

\end{proof}

\section{Correlations versus Entropies}

In this section we present the proof of Lemma \ref{bigentropiesimplycorrelations}, which we restate for convenience of the reader. 
\begin{replem}{bigentropiesimplycorrelations}
\bigentropycor
\end{replem}

\begin{proof}

$\newline$ 

\noindent \textbf{(Part 1)} 

From Lemma \ref{singleshotmerging} in Appendix \ref{singleshotmergingsection} it follows there is a POVM element $M$ and an isometry $V$ such that $\tr(M \rho_A) \geq  1/(2N)$  and 	
\be   \label{distancebound1}
D\left(\frac{(M^{1/2} \otimes V)\ket{\psi}\bra{\psi}_{ACB}(M^{1/2} \otimes V)^{\cal y}}
{\tr(M \rho_A)}, \ket{\Phi}\bra{\Phi}_{A_1C_1} \otimes \ket{\psi}\bra{\psi}_{C'CB}\right) \leq 26\sqrt{\delta},
\ee 
with 
\be
\log(\dim(A_1)) = \lfloor  H_{\min}^{\delta}(A|B) - 2 \log (\delta) - 4 \log(13) \rfloor,
\ee 
and 
\be 
\log(N) = \lfloor I^{\delta}_{\max}(A:B) - 2 \log (\delta) - 4 \log(13) \rfloor.
\ee 
Define 
\be
P_{A_1} := \sum_{k=1}^{\dim(A_1)/2} \ket{k}\bra{k},\quad P_{C_1} := \sum_{k=1}^{\dim(C_1)/2} \ket{k}\bra{k},
\ee
and 
\be 
\tilde \rho_{AC} = \frac{(M^{1/2} \otimes V) \rho_{AC} (M^{1/2} \otimes V)^{\cal y}}{\tr(M \rho_A)}. 
\ee


Let $\mu(dU)$ be the Haar measure on unitaries over a vector space isimorphic to $A_1$ and $C_1$. Then we have
\begin{eqnarray}
\text{Cor}(A:C) &\geq& \int \mu(dU) \tr\left( [(M^{1/2}U P_{A_1} U^{\cal y}M^{1/2})\otimes V^{\cal y} (U^* P_{C_1} U^T \otimes \id_{C'C}) V](\rho_{AC} - \rho_A \otimes \rho_C)  \right) \nonumber \\
&=&  \tr(M \rho_A)  \int \mu(dU)  \tr \left( (U P_{A_1} U^{\cal y} \otimes U^* P_{C_1} U^T ) (\tilde \rho_{A_{1} C_{1}} - \sigma_{A_1} \otimes \tilde \sigma_{C_1}) \right)   \nonumber \\
&=:& M,
\end{eqnarray}
with $\sigma_{A_1}$ and $\sigma_{C_1}$ states on $A_1$ and $C_1$, respectively. We have
\begin{eqnarray}
M = \tr(M \rho_A)   \tr \left( (P_{A_1} \otimes P_{C_1}) (\tilde \Delta(\rho_{A_{1} C_{1}}) - \Delta(\sigma_{A_1} \otimes \tilde \sigma_{C_1})) \right),  
\end{eqnarray}
where the cptp map $\Delta$ is given by
\begin{equation}
\Delta(X) :=  \int \mu(dU) (U \otimes U^*) X (U \otimes U^*)^{\cal y}. 
\end{equation}

By monotonicity of trace-norm under cptp maps and Eq. (\ref{distancebound1}),
\be
D( \Delta(\rho_{A_{1} C_{1}}), \ket{\Phi}\bra{\Phi}_{A_1C_1} ) \leq 26\sqrt{\delta}.
\ee
Moreover since $\sigma_{A_1} \otimes \tilde \sigma_{C_1}$ is a separable state (even product), we have 
\be
\Delta(\sigma_{A_1} \otimes \tilde \sigma_{C_1})) = p \ket{\Phi}\bra{\Phi}_{A_1C_1} + (1-p) \tau_{A_1 C_1},
\ee
with $p \leq 2/|A_1|$. Therefore
\be
M \geq \tr(M \rho_A) \left( 1/256 - 26 \sqrt{\delta}   \right).
\ee


\vspace{0.4 cm}

\noindent \textbf{(Part 2)} 


Part 1 of lemma \ref{lem:H_vs_P} gives projectors $P_A, P_B$ such that $\tr (P_A \rho_A) \geq  1 - 2\nu$, $|P_A| = 2^{H_{\max}^\nu(A)}$
and $\tr (P_B \rho_B) \geq 1-2\delta$, $|P_B|=2^{H_{\max}^\delta(B)}$. Then part 3 of Lemma \ref{lem:H_vs_P} allows us to write $\rho_{AC}$ 
as the following mixture
\be
\rho_{AC}=(1- 2\delta) \pi_{AC} +2\delta \sigma_{AC},
\ee
where $\pi_{AC}$ is the reduced density matrix of the pure state $\ket{\pi}_{ABC}\bra{\pi} := P_B\ket{\psi}\bra{\psi} P_B/\tr (P_B\rho_B)$, and $H_{\max}(\pi_{AC}) = H_{\max}(\pi_B) \leq H^\delta_{\max}(\rho_B)$. Let $Q$ be a projector satisfying $Q\leq P_A$ and chosen from the Haar measure, i.e. $Q = U Q_0 U^{\cal y}$ for a fixed arbitrary projector $Q_0 \leq P_A$ and $U$ is a unitary drawn from the Haar measure in the space in which $P_A$ projects on. Define $\tilde\rho_C := \tr_A (Q\ot I_C \rho_{AC} Q \ot I_C)/\tr(\rho_A Q)$, $\tilde\pi_C := \tr_A (Q\ot I_C \pi_{AC} Q \ot I_C)/\tr(\pi_A Q)$, and $\tilde\sigma_C := \tr_A (Q\ot I_C \sigma_{AC} Q \ot I_C)/\tr(\sigma_A Q)$. Let $\ket{\tilde \pi}\bra{\tilde\pi}_{ABC} := Q\ket{\pi}\bra{\pi} Q/\tr (Q\pi_A)$.

Let us first prove the claim of the lemma under two assumptions and then show how we can relax them. They read
\be
\hspace{0.2 cm} \delta |Q|\geq 4 \sqrt{|P_A|},
\label{eq:assumption1}
\ee
and
\be
\hspace{0.2 cm} H_{\max}^\delta(\rho_B)\geq \log\frac{8}{\nu^2}
\label{eq:assumption2}
\ee

The first step in the proof is to show that the entropy of $\tilde\rho_C$ is bounded from above as follows:
\be
H_{\max}^\eta(\tilde\rho_C)\leq \log |Q| + H_{\max}^\delta(\rho_B)
\label{eq:C_vs_QB}
\ee
where $\eta :=\sqrt{\frac{4\delta}{1-2\nu-2\delta}}$. 
To prove it first note that by subadditivity (Lemma \ref{subadditivityHmax}),
\be
H_{\max}(\tilde \pi_C)\leq  H_{\max}(\tilde \pi_A) +H_{\max}(\tilde \pi_B).
\ee
However $H_{\max}(\tilde \pi_A)\leq \log |Q|$ and
$H_{\max}(\tilde \pi_B)\leq  H_{\max}(\pi_B)\leq H_{\max}^\delta(\rho_B)$ (which follows from the fact that the application of a projection to one system of a bipartite state cannot increase the rank of the other system). Hence we obtain 
\be
H_{\max}(\tilde \pi_C) \leq \log |Q| + H_{\max}^\delta(\rho_B)
\label{eq:pi_vs_rho}
\ee

Applying Lemma \ref{lem:pi_close_rho} it follows that with non-zero probability
\be
D(\tilde\pi_C,\tilde\rho_C)\leq \sqrt{\frac{4\delta}{1-2\nu-2\delta}} = \eta.
\ee
We thus obtain 
\be
H^{\eta}_{\max}(\tilde\rho_C)\leq H_{\max}(\tilde\pi_C),
\ee
which together with Eq. \eqref{eq:pi_vs_rho} gives the required upper bound on 
the entropy of $\tilde\rho_C$ given by Eq.~\eqref{eq:C_vs_QB}. 

Define 
\be \label{eq:C_mu}
\mu := 2 \hspace{0.05 cm} \text{Cor}(A:C) \frac{|P_A|}{|Q|}.
\ee
Then from Lemma \ref{lem:cor-norm}, $D(\tilde\rho_C,\rho_C)\leq \sqrt{\mu}$, and hence
\be
H_{\max}^{\sqrt{\mu}+\eta}(\rho_C)\leq H_{\max}^\eta(\tilde\rho_C).
\label{eq:rho_tilde_rho}
\ee
We now set 
\be
\label{eq:Q_vs_P}
|Q|= \left \lfloor \frac{|P_A|}{2^{3H_{\max}^\delta(B)}} \right \rfloor.
\ee
Recalling that $|P_A|=2^{H_{\max}^\nu(\rho_A)}$, and using  Eqs. \eqref{eq:rho_tilde_rho} and \eqref{eq:C_vs_QB}
we obtain 
\be
H_{\max}^{\sqrt{\mu}+\eta}(\rho_C)\leq H_{\max}^\delta(\rho_B)+H_{\max}^\nu(\rho_A)-3H^\delta_{\max}(\rho_B).
\ee
By subadditivity of max-entropy (Lemma \ref{subHmax}) for system $BC$ with $\epsilon=\frac{\nu}{2}$,
$\epsilon'=\frac{\nu}{2}-2\delta$, and $\epsilon''=\delta$,
\be
H_{\max}^\nu(\rho_A) = H_{\max}^\nu(\rho_{BC}) \leq H_{\max}^{\delta}(\rho_B) + H_{\max}^{\frac{\nu}{2} - 2 \delta}(\rho_C) + \log \frac{8}{\nu^2},
\ee
and so
\be
H_{\max}^{\sqrt{\mu}+\eta}(\rho_C)\leq H^{\frac{\nu}{2}-2\delta}_{\max}(\rho_C)-H_{\max}^\delta(\rho_B)+\log\frac{8}{\nu^2}.
\ee
From Eq. \eqref{eq:assumption2} we obtain $H_{\max}^{\sqrt{\mu}+\eta}(\rho_C)\leq H^{\frac\nu2-2\delta}_{\max}(\rho_C)$ so that 
\be
\label{eq:munudelta}
\sqrt{\mu}+\eta\geq \frac\nu2-2\delta.
\ee
Eqs. \eqref{eq:Q_vs_P} and \eqref{eq:C_mu} give $\text{Cor}(A : C) \geq \mu 2^{-3 H^\delta_{\max}(\rho_B)}$,
which together with Eq. \eqref{eq:munudelta} implies
\be \label{corACunderassump}
\text{Cor}(A:C)\geq \left( \frac\nu2-2\delta - \sqrt{\frac{4\delta}{1-2\nu-2\delta}} \right)^2 2^{-3H^\delta_{\max}(B)}.
\ee

We just derived Eq. (\ref{corACunderassump}) under the two assumptions given by Eqs. (\ref{eq:assumption1}) and (\ref{eq:assumption2}). Let us finally show how we can extend it to a slightly weaker relation without the need of any assumption. Consider the state $\ket{\phi}_{AA':BB':C} := \ket{\psi}_{ABC} \otimes \ket{\Phi}_{A'B'}$, where
$\ket{\Phi}_{A'B'}$ is a maximally entangled state of dimension $|A'| = |B'| = \lfloor 16(\delta^{-2} + \nu^{-2}) \rfloor$. Then the assumptions of Eqs. (\ref{eq:assumption1}) and (\ref{eq:assumption2}) are satisfied. Using Eq. (\ref{corACunderassump}) applied to $\ket{\phi}_{AA':BB':C}$ we get  
\be
\text{Cor}(A:C) = \text{Cor}(AA':C) \geq \left( \frac\nu2-2\delta - \sqrt{\frac{4\delta}{1-2\nu-2\delta}} \right)^2 2^{-3H^\delta_{\max}(BB')},
\ee
 where the first equality follows from the fact that $A'$ is decoupled from $AC$. Finally we note that $H^\delta_{\max}(BB') \leq H^\delta_{\max}(B) + \log|B'| \leq H^\delta_{\max}(B) + \log \left(\delta^{-2} + \nu^{-2} \right) + 5$ and we are done.

\vspace{0.2 cm}

\noindent \textbf{(Part 3)} 

By Lemma \ref{lem:H_vs_P} there is a projector $P$ such that $\tr\left(P\rho_A \right) \geq 1 - 2\delta$ and $|P| = 2^{H_{\max}^{\delta}(A)}$, with $\rho_A$ the reduced state of $\ket{\psi}_{ABC}$. Then applying Lemma \ref{lem:cor-norm} with $M$ equals $P$ we get
\be
\cor(A:C)\geq \left( \frac{1}{2} - \delta \right) D(\tilde \rho_C, \rho_C)^2,
\ee
where $\tilde\rho_C = \tr_{AB} \left(|\tilde\psi\>\<\tilde\psi|_{ABC} \right)$, with
$\ket{\tilde\psi}_{ABC} = \left( P_A \ot \id_{BC} \right) |\psi\>_{ABC} / \sqrt{\tr(P\rho_A)}$ the postselected total state. Thus setting $\gamma := \cor(A:C)^{1/2}(1/2 - \delta)^{-1/2}$, 
\be
H_{\max}^{\gamma}(\rho_C)\leq H_{\max} (\tilde \rho_C)
\label{eq:boost-c}
\ee
By subadditivity of max-entropy (Lemma \ref{subHmax}) we have
\be
H_{\max}^{2\gamma} (\rho_A) = H_{\max}^{2\gamma} (\rho_{BC}) \leq H_{\max}^{\gamma}(\rho_C)+ \log|B| + \log\left( \frac{2}{\gamma^2}  \right),
\ee 
where the entropies are computed with respect to the state $\ket{\psi}_{ABC}$. Again by subadditivity of max-entropy, this time applied to the 
postselected state $\ket{\tilde\psi}_{ABC}$, 
\be
\quad H_{\max}(\tilde\rho_C)\leq H_{\max}(\tilde \rho_A) + \log|B| \leq H_{\max}^{\delta}(\rho_{A}) + \log |B|,
\ee
where we used that $\text{rank}(\tilde\rho_A)\leq 2^{H_{\max}^{\delta}(A)}$, due to the dimension of the projector $P$. The last two inequalities together with 
Eq. \eqref{eq:boost-c} imply Eq. \eqref{eq:boost}.
\end{proof}

The next lemma is used in the proof of Theorem \ref{bigentropiesimplycorrelations}.

\begin{lemma}
\label{lem:pi_close_rho}
Let $\rho_{AC} \in {\cal D}({\cal H}_A \otimes {\cal H}_B)$ be such that
\be
\label{eq:rhoac}
\rho_{AC}= (1-\delta) \pi_{AC} + \delta \sigma_{AC},
\ee
with $0\leq \delta\leq 1$ and $\pi, \sigma \in {\cal D}({\cal H}_A \otimes {\cal H}_B)$. Let $P$ be a projector on system $A$ such that $\tr \left(P \rho_A \right) \geq  1 - 2\nu$, with $0 \leq \nu \leq 1$. Let $Q$ be a random projector given by $U Q_0 U^\dagger$, where $U$ is a Haar distributed unitary on the support of $P$ and $Q_0 \leq P$ is an arbitrary projector. Define $\tilde\rho_C:=\tr_A (Q\ot I_C \rho_{AC} Q \ot I_C)/\tr(\rho_A Q)$, $\tilde\pi_C:=\tr_A (Q\ot I_C \pi_{AC} Q \ot I_C)/\tr(\pi_A Q)$, and $\tilde\sigma_C:=\tr_A (Q\ot I_C \sigma_{AC} Q \ot I_C)/\tr(\sigma_A Q)$. Then with non-zero probability over the choice of $U$,
\be
D(\tilde\rho_C,\tilde\pi_C)\leq \sqrt{\alpha},
\ee
with
\be
\alpha := \frac{\delta\,\frac{|Q_0|}{|P|}+\frac{8}{\sqrt{|P|}}}{(1-2\nu)\frac{|Q_0|}{|P|}-\frac{8}{\sqrt{|P|}}} 
\label{eq:alpha}
\ee
\end{lemma}

\begin{proof}
From Eq. \eqref{eq:rhoac} and the definitions of $\tilde\rho_C,\tilde\pi_C$, and $\tilde\sigma_C$:
\be
\tilde\rho_C=(1-\gamma)\tilde \pi_C + \gamma \tilde\sigma_C,
\ee
with 
\be
\gamma := \delta\frac{\tr (\sigma_A Q)}{\tr(\rho_A Q)}.
\ee
We then have \footnote{Indeed, for normalized pure states $\rho, \sigma$, $D(\rho,\sigma)=\sqrt{1-F^2(\rho,\sigma)}$. We also have $F(\rho,\sigma)=\max|\<\psi|\phi\>|$, with the maximum taken over all $\ket{\psi},\ket{\phi}$ which are purifications of $\rho$ and $\sigma$. Let $\ket{\psi},\ket{\phi}$ be purifications of $\tilde \pi_C$ and $\tilde \sigma_C$, respectively. Then $\ket{\eta} := \sqrt{1 - \gamma} \ket{\psi} \otimes \ket{0} + \sqrt{\gamma} \ket{\phi} \otimes \ket{1}$ is a purification of $\tilde \rho_C$. Therefore $D(\tilde \pi_C,\tilde \rho_C) \leq \sqrt{1-|\braket{\eta}{\psi, 0}|^2}$, giving Eq. (\ref{purifiedrelbufidelityandpuri}).}
\be \label{purifiedrelbufidelityandpuri}
D(\tilde\pi_C,\tilde\rho_C)\leq \sqrt{\gamma}.
\ee

We now show that, with non-zero probability, $\gamma \leq \alpha$ with $\alpha$ given by Eq. \eqref{eq:alpha}.
The function $f(U) := \tr(\rho_A UQ_0 U^{\cal y})$ is 1-Lipschitz continuous on the set of unitaries acting on a space of dimension $|P|$. Hence from Levy's lemma \cite{Led01} we have 
\be
\text{Pr}(|\tr(\sigma_A Q)-\E [\tr (\sigma_A Q)]|\geq \ep)\leq 1- 4 e^{-\frac{(2|P|+1)\ep^2}{16}},
\ee
and similarly
\be
\text{Pr}(|\tr(\rho_A Q)-\E [\tr (\rho_A Q)]|\geq \ep)\leq 1- 4 e^{-\frac{(2|P|+1)\ep^2}{16}}.
\ee
We now take $\ep=8 d^{-\frac12}$ so that both probabilities above are smaller than $1/2$.
Then, since 
\be
\E[\tr(\sigma_A Q)]=\tr (\sigma_A P)\frac{|Q|}{|P|},\quad 
\E[\tr(\rho_A Q)]=\tr (\rho_A P)\frac{|Q|}{|P|},
\ee
by the union bound we have that with non-zero probability:
\begin{eqnarray}
&&\tr(\sigma_A Q)\leq \tr (\sigma_A P)\frac{|Q|}{|P|} + 8 |P|^{-\frac12}\leq \frac{|Q|}{|P|} + 8 |P|^{-\frac12}\nonumber\\
&&\tr(\rho_A Q)\geq \tr (\rho_A P)\frac{|Q|}{|P|} + 8 |P|^{-\frac12}\geq (1-2\nu)\frac{|Q|}{|P|} + 8 |P|^{-\frac12},
\end{eqnarray}
and we are done.
\end{proof}

In the proof of parts 1 and 3 of Lemma \ref{bigentropiesimplycorrelations} we make use of the following simple lemma.

\begin{lemma}
\label{lem:cor-norm}
Given $\rho_{AC} \in {\cal D}({\cal H}_A \otimes {\cal H}_B)$ and an operator $M$ on system $A$ satisfying $0\leq  M \leq \id$, 
\begin{equation}
\label{eq:cor-norm}
\text{Cor}(A:C)\geq \frac{p}{2} D(\tilde \rho_C, \rho_C)^2,
\end{equation}
where $p := \tr (M \rho_A)$ and $\tilde\rho_C :=\tr_A [(\sqrt{M} \ot \id_C) \rho_{AC}(\sqrt{M} \ot \id_C)]/p$
(i.e. $\tilde\rho_{C}$ is postselected state of system $C$ after the measurement on system $A$). 
\end{lemma}

\begin{proof}
We have 
\be
\tr((M \otimes N)(\rho_{AC} - \rho_{A} \otimes \rho_{C}))=
p\tr[N(\tilde\rho_C -\rho_C)].
\ee
Maximizing over $0\leq N \leq 1$ we get 
\begin{equation}
\label{eq:cor-norm}
\text{Cor}(A:C)\geq p D_1(\tilde \rho_C, \rho_C) \geq \frac{p}{2} D(\tilde \rho_C, \rho_C)^ 2,
\end{equation}
where we used Eq. (\ref{variationallytracenorm}) and Lemma \ref{purifiedvstracenorm} in Appendix \ref{purifieddistance}.
\end{proof}

\section{Saturation of max-Mutual Information}

We now turn to the proof of Lemma \ref{saturationmaxmutualinfo}:

\begin{replem}{saturationmaxmutualinfo}
\saturationImax
\end{replem}

\begin{proof}
The first step of the proof is to relate the max-mutual information appearing on the L.H.S. of Eq. (\ref{sublinearmutual}) to an hybrid mutual information, in which  $H^{\delta}_{\min}(X_{C, l}| X_{L, l/2} X_{R, l/2}) $ is replaced by its von Neumann counterpart. We have 
\begin{eqnarray}
I_{\max}^{\delta}(X_{C, l}:X_{L, l/2}X_{R, l/2}) &=& H^{\delta}_{\max}(X_{C, l}) - H^{\delta}_{\min}(X_{C, l}| X_{L, l/2}  X_{R, l/2}) \nonumber \\ &=& H^{\delta}_{\max}(X_{C, l}) + \min_{\sigma} S_{\max}^{\delta}(\rho_{X_{2l}} || \id \otimes \sigma) \nonumber \\ &\leq& H^{\varepsilon}_{\max}(X_{C, l}) + \min_{\sigma}  S_{\max}^{\delta}(\rho_{X_{2l}} || \id \otimes \sigma) ,
\end{eqnarray}
where the last inequality follows from the assumption that $\varepsilon \leq \delta /2$. Here $S_{\max}^{\delta}$ is the smooth max-relative entropy defined in Eqs. (\ref{maxrelent}) and (\ref{smoothmaxrelent}) of Appendix \ref{entropies}.

Using Lemma \ref{distanceextension} of Apendix \ref{entropies}, we can find a $\tilde \rho_{X_{2l}} \in {\cal B}_{\varepsilon}(\rho_{X_{2l}})$ such that $\text{rank}(\tilde \rho_{X_{C, l}}) = 2^{H^{\varepsilon}_{\max}(X_{C, l})}$. Let $P$ be a projector onto its support. Then using $\varepsilon \leq \delta/2$,

\begin{eqnarray}
I_{\max}^{\delta}(X_{C, l}:X_{L, l/2}X_{R, l/2}) &\leq& H^{\varepsilon}_{\max}(X_{C, l}) + \min_{\sigma}  S_{\max}^{\delta/2}( \tilde \rho_{X_{2l}} || P \otimes \sigma) \nonumber \\ &=&  \min_{\sigma}  S_{\max}^{\delta/2}(\tilde \rho_{X_{2l}} || (2^{- H^{\varepsilon}_{\max}(X_{C, l})}P) \otimes \sigma) 
\end{eqnarray}
Using the quantum substate theorem (Lemma \ref{thmsubstate}), a version of Fannes inequality (Lemma \ref{fanaudineq}), and the identity $\min_{\sigma} S(\rho_{AB} || \id \otimes \sigma) = - H(A|B)$ \footnote{It is clear that it sufficies that $\id$ is the identity on the support of $\rho_{A}$.}, we find
\begin{eqnarray}
I_{\max}^{\delta}(X_{C, l}:X_{L, l/2}X_{R, l/2}) &\leq& \frac{2}{\delta} \left( \min_{\sigma}  S(\tilde \rho_{X_{2l}} || (2^{- H^{\varepsilon}_{\max}( X_{C, l} ) } P) \otimes \sigma)      \right) +  \delta  \\ &=& \frac{2}{\delta} \left( H^{\varepsilon}_{\max}(X_{C, l}) + H(\tilde \rho_{X_{L,l/2}X_{R,l/2}}) - H(\tilde \rho_{X_{2l}}) \right) +  \delta \nonumber \\ &\leq& \frac{2}{\delta} \left( H^{\varepsilon}_{\max}(X_{C, l}) + H(X_{L,l/2}X_{R,l/2}) - H(X_{2l}) + 2\varepsilon l + 2 h(\epsilon) \right) +  \delta.  \nonumber
\end{eqnarray}

Thus in order to prove the theorem it suffices to show the existence of an $l$ satisfying $2 \overline{l}_0 (4\xi + 1)^2 \log \left( \frac{2}{\varepsilon} \right) \varepsilon^{-1} \leq l \leq \overline{l}_0 \exp(O(\log(1/\varepsilon)/\varepsilon))$ such that 
\begin{equation} \label{hybridmutualsublinear}
H^{\varepsilon}_{\max}(X_{C, l}) + H(X_{L,l/2}X_{R,l/2}) - H(X_{2l})  \leq \varepsilon l.
\end{equation}

Following an idea of Ref. \cite{Has07}, we prove Eq. (\ref{hybridmutualsublinear}) by contradiction. Suppose that for all $2 \overline{l}_0 (4\xi + 1)^2 \log \left( \frac{2}{\varepsilon} \right) \varepsilon^{-1} \leq l \leq \overline{l}_0 \exp(O(\log(1/\varepsilon)/\varepsilon))$,
\begin{equation} \label{hybridinter00}
H(X_{2l}) \leq H^{\varepsilon}_{\max}(X_{C, l}) + H(X_{L,l/2}X_{R,l/2}) - \varepsilon l .
\end{equation}
Then we will show that this leads to the entropy $H(X_{2l})$ being negative for $l = \overline{l}_0\exp(O(\log(1/\varepsilon)/\varepsilon))$.

Lemma \ref{saturationmutualvonNeumann} gives an argument for the von Neumann mutual information. Thus the challenge is how to adapt the proof of Lemma \ref{saturationmutualvonNeumann} to Eq. (\ref{hybridinter}), in which one of the von Neumann entropies is replaced by the max-entropy. In the sequel we show how to do it exploiting the fact that the state $\rho_{1, ..., n}$ has exponential decay of correlations\footnote{We leave as an open question whether the result holds true for general states without having to assume exponential decay of correlations.}. The idea we explore is to relate $H^{\varepsilon}_{\max}(X_{l})$ to the von Neumann entropy. For general states this is not possible, but by assuming exponential decay of correlations we will achieve it. 

\begin{figure}  
\begin{center}  
\includegraphics[height=4cm,width=14cm,angle=0]{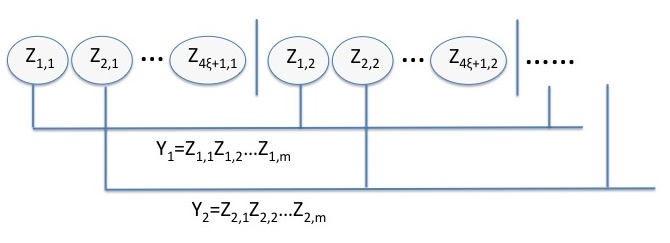}  
\caption{\small Definition of regions $Y$'s and $Z$'s.  \label{fig1}}  
\end{center}  
\end{figure}  

Let us break the region $X_{C, l}$ into a partition of $(4\xi + 1)$ equally sized subregions, labelled by $\{ Y_{k} \}_{k=1}^{(4\xi + 1)}$, each containing $l/(4 \xi + 1)$ sites (with $\xi$ the correlation length of the state). Each region $Y_k$ is given by $m := c/(4 \xi + 1)$ blocks of $l/c$ sites (for a constant $c$ to be determined later), themselves labelled by $Z_{k, j}$. The $Z_{k, j}$ are separated by a distance of $(4 \xi l)/c$ sites to each other, i.e. $Z_{k, j}$ is $(4\xi l) / c$ sites away from $Z_{k, j'}$ for every $j' \neq j$. The structure of the regions is depicted in Figure \ref{fig1}.

By iteratively applying subadditivity of $H_{\max}^{\varepsilon}$ (Lemma \ref{subadditivityHmax}),
\begin{equation} \label{subadditivityHmax2}
H^{\varepsilon}_{\max}(X_{C, l}) \leq \sum_{k=1}^{(4\xi + 1)} H^{ 4^{-(4\xi + 1)} \varepsilon }_{\max}(Y_{k}) + (4 \xi + 1)^2 \log \frac{32}{\varepsilon}.
\end{equation}

We claim that by exponential decay of correlations, for every $k$, the state $\rho_{Y_{k}}$ is close in trace norm to a $m$-fold tensor product of states, each with $l/c$ sites. 
Indeed using the identity 
\begin{eqnarray}
\rho_{Z_{k, 1},...,Z_{k, m}} - \rho_{Z_{k, 1}} \otimes ... \otimes  \rho_{Z_{k, m}} &=& \rho_{Z_{k, 1},...,Z_{k, m}} - \rho_{Z_{k, 1}} \otimes \rho_{Z_{k, 2},...,Z_{k, m}} \nonumber \\ &+&  \rho_{Z_{k, 1}} \otimes \rho_{Z_{k, 2},...,Z_{k, m}} - \rho_{Z_{k, 1}} \otimes \rho_{Z_{k, 2}} \otimes \rho_{Z_{k, 3},...,Z_{k, m}} \nonumber \\ &+& \rho_{Z_{k, 1}} \otimes \rho_{Z_{k, 2}} \otimes \rho_{Z_{k, 3},...,Z_{k, m}} - ... \nonumber \\ &+& ... \nonumber \\  &+& \rho_{Z_{k, 1}} \otimes ... \otimes  \rho_{Z_{k, m-1}, Z_{k, m}}   -  \rho_{Z_{k, 1}} \otimes ... \otimes  \rho_{Z_{k, m}},
\end{eqnarray}
the triangle inequality, and Lemma \ref{bounddatahiding}:
\begin{eqnarray} \label{boundEDC}
\Vert \rho_{Z_{k, 1},...,Z_{k, m}} - \rho_{Z_{k, 1}} \otimes ... \otimes \rho_{Z_{k, m}} \Vert_1 &\leq& \sum_{j=1}^{m} \Vert \rho_{Z_{k, j}, ..., Z_{k, m}} -  \rho_{Z_{k, j}} \otimes  \rho_{Z_{k, j+1},..., Z_{k, m}} \Vert_1 \nonumber \\ &\leq& 2^{\frac{3 l}{c}} \sum_{j=1}^{m}  \text{Cor}\left( Z_{k, j}, (Z_{k, j+1},..., Z_{k, m})\right) \nonumber \\ &\leq& m 2^{\frac{3 l}{c}} 2^{- \frac{4 \xi l} {c \xi}} = m 2^{-\frac{l}{c}},
\end{eqnarray}
where the last inequality follows from $\text{Cor}\left( Z_{k, j}, (Z_{k, j+1},..., Z_{k, m})\right) \leq 2^{- \frac{4 \xi l} {c \xi}}$, which is a consequence of exponential decay of correlations and the fact that $Z_{k, j}$ is $4\xi l/c$ sites away from $Z_{k, j+1},..., Z_{k, m}$. By Lemma \ref{purifiedvstracenorm},
\begin{equation}
D \left( \rho_{Z_{k, 1},...,Z_{k, m}} , \rho_{Z_{k, 1}} \otimes ... \otimes \rho_{Z_{k, m}} \right) \leq \sqrt{2m} 2^{-\frac{l}{2c}}.
\end{equation}

From Eqs. (\ref{subadditivityHmax2}) and (\ref{boundEDC}),
\begin{equation}
H^{\varepsilon}_{\max}(X_{C, l}) \leq \sum_{k=1}^{(4\xi + 1)} H^{\delta}_{\max}( \rho_{Z_{k, 1}} \otimes ... \otimes \rho_{Z_{k, m}}) + (4 \xi + 1)^2 \log \frac{2}{\varepsilon},
\end{equation}
with $\delta := 4^{-(4\xi + 1)} \varepsilon - \sqrt{2m} 2^{-l/(2c)}$. By the quantum equipartition property (Lemma \ref{QEP}), 
\begin{equation}
H^{\delta}_{\max}( \rho_{Z_{k, 1}} \otimes ... \otimes \rho_{Z_{k, m}}) \leq H( \rho_{Z_{k, 1}} \otimes ... \otimes \rho_{Z_{k, m}}) + 10\sqrt{\log(1/\delta)} \sqrt{m} \frac{l}{c}.
\end{equation}
Setting $c := \frac{10^{3}}{\xi^3 \varepsilon^2 \log(1/\varepsilon)}$, we find 
\begin{equation}
H^{\delta}_{\max}( \rho_{Z_{k, 1}} \otimes ... \otimes \rho_{Z_{k, m}}) \leq\sum_{j=1}^{m} H(\rho_{Z_{k, j}}) + \frac{\varepsilon l }{4 (4 \xi + 1)}.
\end{equation}
and thus
\begin{equation} \label{hybridinter11}
H^{\varepsilon}_{\max}(X_{C, l}) \leq \sum_{k=1}^{(4\xi + 1)} \sum_{j=1}^{m} H(Z_{k, j}) - \varepsilon l / 2,
\end{equation}
where we used that $l \geq \overline{l}_0 (4\xi + 1)^2 \log \left( \frac{2}{\varepsilon} \right) \varepsilon^{-1}$.

We can now apply Eq. (\ref{hybridinter11}) to Eq. (\ref{hybridinter00}), and apply subadditivity of von Neumann entropy repeatedly to the region $X_{L,l/2}X_{R,l/2}$ until we reach regions of size $l/c$ to obtain:
\begin{equation} 
H(X_{2l}) \leq  2c H(X_{l/c}) - \frac{\varepsilon}{2} l.
\end{equation}
Appplying the equation above recursively $\log_c(l/\overline{l}_0)$ times we get
\begin{equation} 
H(X_{2l}) \leq  \frac{2l}{\overline{l}_0} H(X_{\overline{l}_0}) - \frac{\varepsilon}{2}l \log_{c}(l/\overline{l}_0),
\end{equation}
which leads to a contradiction for $l = \overline{l}_0 c^{2/\varepsilon} = \overline{l}_0 \exp \left( O \left( \log(1/\epsilon) / \epsilon \right)   \right)$.
\end{proof}


\section{Proof of Theorem \ref{maintheoremmixed}}

In this section we prove Theorem \ref{maintheoremmixed}. The main idea is to consider a purification of the state and follow the proof of the pure state case, adding the purifying system to the part attributed to the reference system, or environment system, in each of the two applications of the entanglement distillation protocol. Apart from this, the proof will follow closely the argument in the pure state case. 

\begin{repthm}{maintheoremmixed}
\entropyboundmixed
\end{repthm}

\begin{proof}
Let $\ket{\psi}_{SE}$ be a purification of the state $\rho$, i.e. $\tr_{R}(\ket{\psi}_{SE}\bra{\psi}) = \rho$. We start by 
applying Lemma \ref{saturationmaxmutualinfo} twice to the two boundaries of the region $X$, with $\delta := 10^{-8}$, $\varepsilon := \min \left(1/(10^5 \xi), \delta/8 \right)$, and $\overline{l_0} := \max \left(4(1+H_{\max}(\rho))  l_0, 4(1+H_{\max}(\rho)) c l_0  \varepsilon/(\xi \log(2/\varepsilon)) \right)$, for a constant $c > 0$. Then we find there are regions $Y$ and $\tilde Y$ (see Fig. \ref{fig2} (a)), each of size $2l$ with $O(\xi^2 \log \left( 2/\varepsilon\right))  \leq l/\overline{l}_0 \leq \exp(O(\log(1/\varepsilon)/\varepsilon))$, and at a distance at most $\overline{l}_0\exp(O(\log(1/\varepsilon)/\varepsilon))$ sites away from the two boundaries of $X$, respectively, such 
that 
\begin{equation}  \label{saturationmutualfirstpart2}
I^{\delta/4}_{\max}(Y_{C, l} : Y_{L, l/2}  Y_{R, l/2}) = H^{\delta/4}_{\max}(Y_{C, l}) - H^{\delta/4}_{\min}(Y_{C, l}| Y_{L, l/2}  Y_{R, l/2}) \leq \frac{l}{8 \xi},
\end{equation}
and likewise for $\tilde Y$.  We note that from the choice of parameters if follows that $l \geq 300 \xi$. By sub-additivity we then find
\begin{equation} \label{saturationmutualfirstpart22}
I^{\delta}_{\max}(Y_{C, l} : Y_{L, l/2}  Y_{R, l/2}E) \leq \frac{l}{8 \xi} + H_{\max}(E).
\end{equation}

We now argue that exponential decay of correlations implies that the entropy of $\rho_{Y_{C, l}}$ must be small. Let $R$ be the complementary region to $Y$ (see Fig. \ref{fig2} (b)). Then from exponential decay of correlations we have $\text{Cor}(Y_{C, l} : R) \leq 2^{- l / (2\xi)}$. Applying part (1) of Lemma \ref{bigentropiesimplycorrelations} we find that
\begin{equation} \label{Hmaxeqbound22}
H_{\min}^{\delta/4}(Y_{C, l}| Y_{L, l} Y_{R, l} ) < -2 \log(\delta) + 9, 
\end{equation}
since otherwise
\begin{equation}
\text{Cor}(Y_{C, l} : R)  \geq \left(\frac{1}{256} - 26 \sqrt{\delta}\right) 2^{- I^{\delta/4}_{\max}(Y_{C, l} : Y_{L, l/2}  Y_{R, l/2}E) - 2 \log(\delta) + 9 } \geq 2^{- \left( \frac{l}{8 \xi} - 34 \right)},
\end{equation}
in contradiction with the correlation length being $\xi$, since $l \geq 300 \xi$. 

From Eq. (\ref{Hmaxeqbound22}) and Eq. (\ref{saturationmutualfirstpart22}) we thus have
\begin{equation}  \label{boundYc22}
H^{\delta/4}_{\max}(Y_{C, l}) \leq \frac{l}{8 \xi} -2 \log(\delta) + 9 + H_{\max}(\rho),
\end{equation}
where we used $H(E) = H(\rho)$, since $E$ is purification of the state.

Applying the same reasoning to $\tilde Y$ we find that also $H^{\delta/4}_{\max}(\tilde Y_{C, l}) \leq \frac{l}{8 \xi} - 2 \log(\delta) + 9 + H_{\max}(\rho)$. 

Define $A$ as the region between $Y_{C, l}$ and $\tilde Y_{C, l}$, $B := Y_{C, l} \tilde Y_{C, l}$, and $C$ the complementary region to $AB$ (see Fig. \ref{fig2} (c)). Note that $A$ differs from $X$ by at most $2 \overline{l}_0\exp(O(\log(1/\varepsilon)/\varepsilon))$ sites and so by Lemma \ref{subadditivityHmax}, for every $\nu > 0$,
\begin{equation} \label{xregionboundbyA22}
H_{\max}^{2\nu}(X) \leq H_{\max}^{\nu}(A) + 2\overline{l}_0\exp(O(\log(1/\varepsilon)/\varepsilon)) + \log \frac{2}{\nu^2}.
\end{equation}

We now prove an upper bound on $H_{\max}^{\nu}(A)$, with $\nu = 0.01$. Since $A$ and $C$ are separated by $l$ sites, we have from exponential decay of correlations that 
\begin{equation} \label{corrlargesecondpart22}
\text{Cor}(A:C) \leq 2^{-l/\xi}. 
\end{equation}
From part (2) of Lemma \ref{bigentropiesimplycorrelations} and Eq. (\ref{corrlargesecondpart22}) we then find that
\begin{equation} \label{BonC22}
H_{\max}^{\nu}(C) < 3 H_{\max}^{\delta}(BE),
\end{equation}
since otherwise, using that $l \geq 300\xi$,
\begin{equation}
\text{Cor}\left( A:C \right) \geq \left( \frac\nu2-2\delta - \sqrt{\frac{4\delta}{1-2\nu-2\delta}} \right)^2 2^{-3\left( H^\delta_{\max}(BE) + 16 \log \left(\delta^{-2} + \nu^{-2} \right) + 5 \right)} \geq  2^{-3l/(4 \xi)},
\end{equation}
contradicting Eq. (\ref{corrlargesecondpart22}). In the last inequality of the equation above we used that
\begin{eqnarray}
H_{\max}^{\delta}(BE) &=& H_{\max}^{\delta}(Y_{C, l} \tilde Y_{C, l}E) \nonumber \\ &\leq& H_{\max}(\rho) + H_{\max}^{\delta/4}(Y_{C, l}) + H_{\max}^{\delta/4}(\tilde Y_{C, l}) + \log \frac{8}{\delta^2} \nonumber \\ &\leq& 3H_{\max}(\rho) +  \frac{l}{4 \xi} - 4 \log(\delta) + 18  + \log \frac{8}{\delta^2},
\end{eqnarray}
where the first inequality follows from subadditivity of the max-entropy (Lemma \ref{subadditivityHmax}) and the second from Eq. (\ref{boundYc22}).

Using subadditivity of the max-entropy (Lemma \ref{subadditivityHmax}) again, we get
\begin{eqnarray}
H_{\max}^{3\nu + \delta}(A) = H_{\max}^{3\nu + \delta}(BCE) &\leq& H_{\max}(\rho) +  H_{\max}^{\delta}(B) +  H_{\max}^{\nu}(C) + \log \frac{2}{\nu^2} \nonumber \\ &\leq& 4H_{\max}^{\delta}(B) + \log \frac{2}{\nu^2} \nonumber \\ &\leq&  2\overline{l}_0\exp(O(\log(1/\varepsilon)/\varepsilon)),
\end{eqnarray}
where we used Eq. (\ref{BonC22}) and that $H_{\max}(B) \leq 2l \leq \overline{l}_0\exp(O(\log(1/\varepsilon)/\varepsilon))$. Therefore by Eq. (\ref{xregionboundbyA22}),
\begin{equation}
H_{\max}^{6\nu + 2\delta}(X) \leq 2\overline{l}_0\exp(O(\log(1/\varepsilon)/\varepsilon)) = O(l_0)\exp(O(\log(1/\varepsilon)/\varepsilon)).
\end{equation}

This is already an area law for $X$, although with a fixed error (recall that we fixed $\delta=10^{-8}$ and $\nu=0.01$). To finish the proof we show how applying exponential decay of correlations once more we can reduce the error. Let $Z$ be a region of size $2l$ separating $X$ from the remaining sites by a distance $l$, and denote by $U$ the complementary region to $XZ$ (see Fig. \ref{fig2} (d)). By exponential decay of correlations we have $\text{Cor}(X:U) \leq 2^{-l/\xi}$. Using part (3) of Lemma \ref{bigentropiesimplycorrelations} we then get
\begin{equation}
H_{\max}^{2^{-l/(4\xi)}}(X)  \leq H_{\max}^{2^{-l/(2\xi) + 1}}(X) \leq H_{\max}^{6\nu + 2\delta}(X) + 2l  \leq O(l_0) \exp(O(\log(1/\varepsilon)/\varepsilon)) + 2l,
\end{equation}
where the first inequality follows assuming $l \geq 4\xi$. 
\end{proof}

\section{Proofs of Corollaries \ref{vonNeumannarealaw} and \ref{MPS}}

\noindent \textit{Proof Corollary \ref{vonNeumannarealaw}:}

The proof follows from Lemma 2 of \cite{Has07}. Let $\{ \lambda_i \}$ be the eigenvalues of $\rho_X := \tr_{\backslash X} (\ket{\psi} \bra{\psi})$ in decreasing order. Theorem \ref{maintheorempure} gives that for every $l$,
\be \label{constraints}
\sum_{i \hspace{0.01 cm} : \hspace{0.01 cm}  i \geq 2^{c' l_0 \exp(c \log(\xi) \xi) + l}} \lambda_i \leq 2^{- l / 8 \xi}. 
\ee
Eq. (\ref{vnbound}) follows by maximizing $- \sum_i \lambda_i \log(\lambda_i)$ subject to the constraints of Eq. (\ref{constraints}) (see Lemma 2 of \cite{Has07}). 

\noindent \textit{Proof Corollary \ref{MPS}:}

The corollary follows from the argument of \cite{Has07} (section "Matrix Product States").

\section{Conclusions and Open Questions}

In this work we proved that for one-dimensional quantum states an area law for their entanglement entropy follows merely from the fact that the state has exponential decay of correlations. While intuitively very natural, the relation of exponential decay of correlations and area law was put into question by the peculiar kind of correlations embodied in the so-called quantum data hiding states. The results of the paper thus shows that, despite the difficulties caused by these type of correlations, the physically motivated intuition is nonetheless correct. 

In a sense the obstruction provided by ideas from quantum information theory, namely the concept of data hiding states, can be overcome by considering the problem \textit{also} from the perceptive of quantum information theory. In particular we employed the central idea in quantum Shannon theory of \textit{decoupling} two quantum systems by performing a random measurement in one of them -- an idea that has been used to derive the best known protocols for a variety of information theoretic tasks \cite{ADHW09, HOW05, HOW07, DBWR10, BCR11, BDHSW09, HHYW08}) -- as well as recent developments in the framework of single-shot quantum information theory \cite{Ren05, Tom12}. In this respect the results of this paper represent an interesting application of this framework to a problem \textit{outside} of information theory, giving further evidence of its utility (see e.g. Refs. \cite{dRARDV11, HO11} for other examples).

We now list a few open problems:

\begin{enumerate}
\item Can we improve the dependency of correlation length in the entanglement entropy bound? The exponential dependency found here seems hardly optimal, and for groundstates of 1D gapped local Hamiltonians -- an important class of states with exponential decay of correlations -- the recent result of Arad, Kitaev, Landau and Vazirani \cite{AKLV12} shows that such an improvement is indeed feasible. A possible direction to get a sharper bound would be to improve the result about saturation of mutual information, getting a better bound on the size of the region one must vary in order to get small mutual information, perhaps exploiting the fact that the state under consideration has exponential decay of correlations.

\item How small can we choose the constant the smoothing term in Eq. (\ref{entropybound}) of Theorem \ref{maintheorempure}? We believe this is an interesting question because an improvement from $l/8\xi$ to any number strictly smaller than one would imply a matrix product representation of \textit{sublinear} bond dimension, which by the methods of Refs. \cite{AAI10, SC10} would lead to a subexponential-time algorithm for obtaining the ground state of gapped 1D Hamiltonians. 

\item Can we prove an extension of Theorem \ref{saturationmaxmutualinfo} for general states, without the assumption of exponential decay of correlations? Hastings' result about the saturation of mutual information \cite{Has07} (see Lemma \ref{saturationmutualvonNeumann}) and its single-shot counterpart given by Lemma \ref{saturationmaxmutualinfo} appear to be powerful results about the distribution of correlations on different length-scales. Are there more applications of them?

\item Can we extend the result to mixed states with exponential decay of correlations? Theorem \ref{maintheoremmixed} gives a first step in this direction, however the result is not completely satisfactory as the statement is only meaningful for states of low entropy. In turn, it would be very interesting to explore whether for general mixed states exponential decay of correlations implies a bound on the \textit{mutual information} of an arbitrary region with its complementary region (of the same flavour as the area law for thermal states proved in \cite{WFHC08}). We believe the techniques developed in this paper might be useful in addressing this problem, although a new idea appears to be needed. 


\item Can we prove area law from exponential decay of correlations for states of an infinite number of particles? The challenge here is how to generalize the information-theoretic tools employed in the proof, in particular the entanglement distillation protocol, to the setting of von Neumann algebras. In this respect the recent results of Ref. \cite{BFS11} might be useful. 

\item Finally can the approach of this paper be used to prove an area law from exponential decay of correlations in higher dimensions, or even \textit{just} an area law for groundstates of gapped local Hamiltonians in 2D? The latter is one of the most important open problems in the field of quantum Hamiltonian complexity \cite{Osb11} and we hope this work will drawn the attention of quantum information theorists to it. Again a new idea appears to be needed, the main difficulty being that in higher dimensions the separation distance between two regions is not of the same order as the number of sites in the separating region, a feature that was crucially explored in this work. 

\end{enumerate}

\section{Acknowledgement}

We would like to thank Dorit Aharonov, Itai Arad, and Aram Harrow for interesting discussions on area laws and related subjects, Matt Hastings for useful correspondence and Milan Holz\"apfel for useful comments on an earlier version of the paper. FB acknowledges support from the Swiss National Science Foundation, via the National Centre of Competence in Research QSIT. MH thanks the support of  EC IP QESSENCE,  ERC QOLAPS, and National Science Centre, grant no. DEC-2011/02/A/ST2/00305. Part of this work was done at National Quantum Information Centre of Gdansk. F.B. and M.H. thank the hospitality of Institute Mittag Leﬄer within the program ”Quantum Information Science” (2010), where part of this work was done. 

\appendix

\section{Purified Distance} \label{purifieddistance}

Let ${\cal H}$ be a finite dimensional vector space and ${\cal B}({\cal H})$ the set of linear operators on ${\cal H}$. We define set of sunormalized states ${\cal D}_{\leq}({\cal H}) = \{ \rho \in {\cal B}({\cal H}) : \tr(\rho) \leq 1\}$ and the set of normalized states ${\cal D}({\cal H}) = \{ \rho \in {\cal B}({\cal H}) : \tr(\rho) = 1\}$. 

Let $F(\rho, \sigma) = \tr(\sqrt{\sigma^{1/2}\rho \sigma^{1/2}})$ be the fidelity of $\rho$ and $\sigma$. We quantify the distance of quantum states by the purified distance \cite{TCR09}:
\begin{definition}
Let $\rho, \sigma \in {\cal D}_{\leq}({\cal H})$. The purified distance between $\rho$ and $\sigma$ is defined as
\begin{equation}
D(\rho, \sigma) = \sqrt{1 - \overline{F}(\rho, \sigma)^{2}},
\end{equation}
where $\overline{F}(\rho, \sigma) = F(\rho, \sigma) + \sqrt{(1 - \tr(\rho))(1 - \tr(\sigma))}$ denotes the generalized fidelity. 
\end{definition}
We define the $\varepsilon$-ball around $\rho$ as 
\begin{equation}
{\cal B}^{\varepsilon}(\rho) = \{ \rho' \in {\cal D}_{\leq}({\cal H}) : D(\rho, \rho') \leq \varepsilon \}.
\end{equation}

The next lemma is a slight variant of Uhlmann's theorem for the fidelity. 

\begin{lemma}[Uhlmann's Theorem for Purified Distance \cite{TCR09}] \label{purifieddistancepurification}
Let $\rho, \sigma \in {\cal D}_{\leq}({\cal H})$. Then
\begin{equation}
D(\rho, \sigma) = \min_{\ket{\psi}, \ket{\phi}} D(\ket{\psi}\bra{\psi}, \ket{\phi}\bra{\phi}),
\end{equation}
where the minimum is taken over purifications $\ket{\psi}, \ket{\phi}$ of $\rho$ and $\sigma$, respectively.
\end{lemma}

For $\rho, \sigma \in {\cal D}_{\leq}({\cal H})$, define 
\begin{equation} \label{tracenormdef}
D_{1}(\rho, \sigma) := \frac{1}{2}\Vert \rho - \sigma \Vert_1  +  \frac{1}{2}| \tr(\rho) - \tr(\sigma) |.
\end{equation}
In Ref. \cite{TCR09} it is shown that
\begin{equation}  \label{variationallytracenorm}
D_{1}(\rho, \sigma) = \max_{0 \leq M \leq \id} |\tr(M(\rho - \sigma))|.
\end{equation}

\begin{lemma}[Purified Distance Versus Trace Norm \cite{TCR09}] \label{purifiedvstracenorm}
For $\rho, \sigma \in {\cal D}_{\leq}({\cal H})$,
\begin{equation}
D_{1}(\rho, \sigma) \leq D(\rho, \sigma) \leq \sqrt{2  D_{1}(\rho, \sigma) }.
\end{equation}
\end{lemma}

One of the main reasons why it is useful to consider the purified distance is the following:

\begin{lemma} [Purified Distance of Extensions \cite{TCR09}] \label{distanceextension}
Let $\rho, \sigma \in {\cal D}_{\leq}({\cal H})$ and $\tilde \rho \in  {\cal D}_{\leq}({\cal H} \otimes {\cal H}')$ be an extension of $\rho$. Then there exists an extension $\tilde \sigma \in  {\cal D}_{\leq}({\cal H} \otimes {\cal H}')$ of sigma satisfying $D(\rho, \sigma) = D(\tilde \rho, \tilde \sigma)$.
\end{lemma}

\section{Properties of Entropies} \label{entropies}

\begin{lemma}
\label{lem:H_vs_P}
Let $ \rho \in {\cal D}_{\leq}({\cal H})$. Then,
\begin{enumerate}
\item There exists a projector $P$ such that $\log |P| = H_{\max}^\delta(\rho)$ and $\tr(P\rho) \geq 1 - 2\delta$. 
\item If a projector $P$ satisfies $\tr(P\rho) \geq 1-\delta$, then $H^{\sqrt{\delta}}_{\max}(\rho)\leq \log |P|$.
\item Let $\{ \lambda_k \}$ be the eigenvectors of $\rho$ in decreasing order. Then $\sum_{k=1}^{2^{H_{\max}^\delta(\rho)}}\lambda_k \geq 1 - 2 \delta$
\end{enumerate}
\end{lemma}

\begin{proof}
Let $\rho_{\delta}$ be such that $D(\rho,\sigma)\leq \delta$ and $H_{\max}^\delta(\rho)=H_{\max}(\rho_{\delta})$. Then by Lemma \ref{purifiedvstracenorm}, 
$D_1(\rho,\sigma)\leq \delta$ and since $\tr(\rho) = 1$,  $\tr(\rho_{\delta}) \geq 1-\delta$. Now, let $P$ be the projector onto the support of $\rho_{\delta}$, which satisfies $\log|P| = H_{\max}^\delta(\rho)$. We then obtain 
\be
|\tr (P\rho) - \tr(\sigma)| = |\tr (P(\rho - \sigma))| \leq D_{1}(\rho, \sigma) \leq \delta. 
\ee
Therefore $\tr (P\rho) \geq 1-2 \delta$, proving part 1 of the lemma.

To prove part 2, it is enough to take $\rho_{\delta} := P\rho P$ and note that $F(\rho, \rho_{\delta}) = \sqrt{\tr(P \rho)}$ for normalized $\rho$, and so $D(\rho, P\rho P) \leq \sqrt{1 - \tr (P\rho)} \leq \sqrt{\delta}$.

Finally, the third part is a consequence of the statement of the first part and the relation
\begin{equation}
\max_{Q} \{ \tr(Q\rho) : Q^2 = Q, \rank(Q) = q \} = \sum_{k=1}^q \lambda_k.
\end{equation}
\end{proof}

Following Ref. \cite{Dat09} we define the max-relative entropy as
\begin{equation} \label{maxrelent}
S_{\max}(\rho || \sigma) := \min \{ \lambda : \rho \leq 2^{\lambda} \sigma   \}
\end{equation}
and the $\varepsilon$-smooth relative entropy as
\begin{equation} \label{smoothmaxrelent}
S_{\max}^{\varepsilon}(\rho || \sigma) := \min_{\tilde \rho \in {\cal B}_{\varepsilon}(\rho)} S_{\max}(\tilde \rho || \sigma).
\end{equation}

The next lemma was first proven in Ref. \cite{JJS07} (see also Ref. \cite{JN11}) and is a fundamental piece in the proof of Theorem \ref{saturationmaxmutualinfo}.
\begin{lemma}[Quantum Substate Theorem \cite{JN11}] 
Let $\rho, \sigma \in {\cal D}(\mathbb{C}^d)$ be such that $\text{supp}(\rho) \subseteq \text{supp}(\sigma)$. For any $\varepsilon \in (0, 1)$, 
\begin{equation}
S_{\max}^{\varepsilon}(\rho || \sigma) \leq \frac{S(\rho || \sigma)}{\varepsilon} + \log \frac{1}{1 - \varepsilon}.
\end{equation}
\label{thmsubstate}
\end{lemma}

A version of next lemma first appeared in \cite{Has07} (see also \cite{AALV10}) where it was used to prove an area law for groundstates of 1D gapped Hamiltonians.
\begin{lemma}[Saturation of Mutual Information] \label{saturationmutualvonNeumann}
Let $\rho_{1,...,n} \in {\cal D}((\mathbb{C}^2)^{\otimes n})$ and $s \in [n]$ be a particular site. Then for all $ \varepsilon > 0$ and $l_0 > 0$ there is a $l$ satisfying $1 \leq l/l_0 \leq \exp(O(1/\varepsilon))$ and a connected region $X_{2l} := X_{L, l/2} X_{C, l} X_{R, l/2}$ of $2l$ sites (the borders $X_{L, l/2}$ and $X_{R, l/2}$ with $l/2$ sites each, and the central region $X_{C, l}$ with $l$ sites) centred at most $l_0\exp(O(1/\varepsilon))$ sites away from $s$ (see Fig. \ref{sat}) such that
\begin{equation} \label{sublinearmutualvonNeumann}
I(X_{C, l}: X_{L, l/2}X_{R, l/2}) \leq \varepsilon l.
\end{equation}
\end{lemma}
\begin{proof}
We prove Eq. (\ref{sublinearmutualvonNeumann}) by contradiction. Suppose that for all $ l_0 \leq l \leq l_0 \exp(O(\log(1/\varepsilon))$ 
\begin{equation} \label{hybridinter}
H(X_{2l}) \leq H(X_{C, l}) + H(X_{L,l/2}X_{R,l/2}) - \varepsilon l .
\end{equation}
Then we show that this leads to the entropy $H(X_{2l})$ being negative for $l = l_0\exp(O(1/\varepsilon))$. Using subadditivity of entropy 
\begin{equation} \label{recur1}
H(X_{2l}) \leq H(X_{C, L, l/2}) + H(X_{C, R, l/2}) + H(X_{L,l/2}) + H(X_{R,l/2}) - \varepsilon l,
\end{equation}
where $X_{C, L, l/2}$ and $X_{C, R, l/2}$ are two regions such that $X_{C, l} = X_{C, L, l/2}X_{C, R, l/2}$. Thus we have 
\be
H(X_{2l}) \leq 4 H(X_{l/2}) -\epsilon l
\ee
We can now apply Eq. (\ref{hybridinter}) to $H(X_{l/2}) $. Doing so recursively, until we reach regions of size $l_0$, we get
\begin{equation}
H(X_{2l}) \leq l - \log_4 \left( l/ l_0 \right) \varepsilon l, 
\end{equation}
which leads to a contradiction choosing $l = l_0\exp(O(1/\varepsilon))$. 
\end{proof}
The difference of Lemma \ref{saturationmutualvonNeumann} with the result of \cite{Has07, AALV10} is that the latter is concerned with the mutual information between two equally sized neighbouring regions $X_{L, l}X_{R, l}$. Note however that the proofs are very similar in both cases. 

The next lemma is a version of Fannes inequality \cite{Fannes} for the purified distance, which follows directly from the version of the inequality proved in \cite{Aud06} and Lemma \ref{purifiedvstracenorm}.
\begin{lemma}[Fannes-type Inequality \cite{Aud06}] \label{fanaudineq}
Let $\rho, \sigma \in {\cal D}(\mathbb{C}^d)$. Then
\begin{eqnarray}
|H(\rho) - H(\sigma)| &\leq& \log(d-1) D(\rho, \sigma) + h( D(\rho, \sigma) )
\end{eqnarray}
with $h(\varepsilon) := - \varepsilon \log(\varepsilon) - (1 - \varepsilon)\log(1 - \varepsilon)$ the binary entropy.
\end{lemma}

In proof of the main theorems we also make repeatedly use of the following version of the subadditivity inequality for max-entropy, which follows directly from a chain rule relation of Ref. \cite{Tom12}. 
\begin{lemma}[Subadditivity Smooth max-Entropy] \label{subadditivityHmax}
Given a state $\rho_{AB}$ and $\varepsilon, \varepsilon', \varepsilon'' > 0$,
\begin{equation} 
H_{\max}^{\varepsilon + \varepsilon'+ 2 \varepsilon''}(AB) \leq H_{\max}^{\varepsilon'}(A) + H_{\max}^{\varepsilon''}(B) + \log \frac{2}{\varepsilon^2(1 - \varepsilon - \varepsilon' - 2\varepsilon'')}.
\label{subad-1}
\end{equation}
Moreover we have
\be
H_{\max}(AB) \leq H_{\max}(A) + H_{\max}(B)
\label{subad-2}
\ee
\label{subHmax}
\end{lemma}
\begin{proof} We first prove the formula \eqref{subad-1}. Eq. (5.12) of Ref. \cite{Tom12} with a trivial $C$ system gives 
\begin{equation}
H_{1/2}^{\varepsilon + \varepsilon' + 2 \varepsilon''}(AB) \leq H_{1/2}^{\varepsilon'}(A|B) + H_{1/2}^{\varepsilon''}(B) + \log \frac{2}{\varepsilon^2}.
\end{equation}
Using the inequality
\be
H_{1/2}^{\varepsilon'}(A|B) \leq H_{1/2}^{\varepsilon'}(A).
\ee
we obtain 
\begin{equation}
H_{1/2}^{\varepsilon + \varepsilon' + 2 \varepsilon''}(AB) \leq H_{1/2}^{\varepsilon'}(A) + H_{1/2}^{\varepsilon''}(B) + \log \frac{2}{\varepsilon^2}.
\end{equation}

In the case of zero smoothing $\varepsilon' =  0$, this inequality follows directly from the relation $H_{1/2}(A|B) = \max_{\sigma} \log F(\rho_{AB}, \id_A \otimes \sigma_B)^2$ with $\id_A$ the identity onto the support of $\rho_A$. For $\varepsilon' > 0$, let $\tilde \rho_{AB} \in {\cal B}_{\varepsilon'}(\rho_{AB})$ be such that $\text{rank}(\tilde \rho_{A}) = 2^{H_{1/2}^{\varepsilon'}(A)}$. Then
\begin{equation}
H_{1/2}^{\varepsilon'}(A|B)_{\rho} = H_{1/2}(A|B)_{\tilde \rho} \leq H_{1/2}(\tilde \rho_A) = H_{1/2}^{\varepsilon'}(A). 
\end{equation}

Eq. (\ref{subad-1}) then follows from the relation
\begin{equation}
H_{\max}^{\varepsilon}(X) \leq H_{1/2}^{\varepsilon}(X) - \log(1 - \varepsilon).
\end{equation}
from Lemma 4.1 of \cite{RW04g}.

The formula \eqref{subad-1} just says that the rank of a state is no greater than the product of ranks 
of its subsystems.
\end{proof}

The next lemma is from Refs. \cite{TCR09, Tom12}  is used in the proof of Lemma \ref{saturationmaxmutualinfo} to relate the max-entropy and the von Neumann entropy. We actually need a slight generalization of the original formulation in which the systems are not necessarily identical. The original proof however carries through to this case without any modification.
\begin{lemma}[Quantum Equipartition Property \cite{TCR09, Tom12}]
Let $\pi_{1,...,n} = \pi_{1} \otimes ... \otimes \pi_{n}$, with $\pi_{k} \in {\cal D}(\mathbb{C}^{d})$ for all $k \in [n]$. Then
\begin{equation}
\frac{1}{n}H_{\max}^{\varepsilon}(\pi_{1,..., n}) \leq \frac{1}{n} H(\pi_{1,...,n}) + 4d \sqrt{\frac{\log(2/\varepsilon^2)}{n}}.
\end{equation}
\label{QEP}
\end{lemma}

\section{Bound on Data Hiding} 

The next lemma gives a limit to quantum data hidding in terms of the minimal local dimension of a bipartite operator. 
\begin{lemma} \label{bounddatahiding}
For every $L \in \mathbb{B}(\mathbb{C}^{d} \otimes \mathbb{C}^{D})$, with $d \leq D$,
\begin{equation}
\Vert L \Vert_{1} \leq d^2\max_{\Vert X \Vert \leq 1, \Vert Y \Vert \leq 1} | \tr( (X \otimes Y) L ) |.
\end{equation}
\end{lemma}

\begin{proof}
From the variational characterization of the trace norm there is a $M$ with $\Vert M \Vert \leq 1$ such that $\Vert L \Vert_{1} = \tr(ML)$. As will be shown below, we can write
\begin{equation}
M = \sum_{k=1}^{d^2} X_k \otimes Y_k,
\label{eq:MXY}
\end{equation}
with $\Vert X_k \Vert \leq 1$, $\Vert Y_k \Vert \leq 1$. Then we have 
\begin{equation}
\Vert L \Vert_{1} = \tr(ML) = \sum_{k=1}^{d^2} \tr((X_k \otimes Y_k) L) \leq d^2\max_{\Vert X \Vert \leq 1, \Vert Y \Vert \leq 1} | \tr( (X \otimes Y) L ) |.
\end{equation}
We now show  \eqref{eq:MXY}. To this end
we write $M$ in block form:
\be
M = \sum_{a=1}^d \sum_{b=1}^d |a\>\<b| \otimes M_{a, b}
\ee
where $M_{a,b}$ are $D\times D$ matrices.
It suffices to prove that $|| M_{a, b} || \leq 1$. For this
we use that for every operator $X$,
\be
||X|| = \max_{|\psi\>, |\phi\>} \<\psi| X |\phi\>.
\ee
Then
\ben
&& 1 \geq ||M|| =\max_{|\psi\>, |\phi\>} \<\psi| M |\phi\> 
\geq \max_{ |a, \psi'\>, |b, \phi'\>} \<a, \psi'| M |b, \phi'\> \nonumber \\
&& = \max_{a,b}\max_{ |\psi'\>, |\phi'\> } \<\psi'| M_{a, b} |\phi'\> = \max_{a,b}|| M_{a, b} ||.
\een
\end{proof}

\section{Decoupling in Haar Random States} \label{decouplingHaar}

The objective of this section is to state and prove the following well-known result (see e.g. \cite{HLW06}):

\begin{lemma}
Let $\ket{\psi}_{AB}$ be a Haar random state in ${\cal H}_A \otimes {\cal H}_B$ with $|A| \geq |B|$. Then
\begin{equation}
\mathbb{E} \left( D(\rho_B, \tau_B) \right) \leq \left( 2\frac{|B|}{|A|} \right)^{1/4},
\end{equation}
where $\rho_B$ is the $B$ reduced density matrix of $\ket{\psi}_{AB}$ and $\tau_B$ is the maximally mixed state on ${\cal H}_B$.
\end{lemma}

\begin{proof}
Lemma V.3 of Ref. \cite{BH09} gives 
\begin{equation}
\mathbb{E} \left( \Vert \rho_B - \tau_B  \Vert_2   \right) \leq \frac{1}{|A|^{1/2}},
\end{equation}
and since for $d \times d$ matrix $X$, $\Vert X \Vert_1 \leq \sqrt{d} \Vert X \Vert_2$,
\begin{equation}
\mathbb{E} \left( \Vert \rho_B - \tau_B  \Vert_1   \right) \leq \sqrt{\frac{|B|}{|A|}}.
\end{equation}
The statement then follows from Lemma \ref{purifiedvstracenorm}.
\end{proof}

\section{Decay of Correlations in MPS} \label{DecayMPS}

Given a quantum channel $\Lambda(\rho) = \sum_{k} A_k \rho A_k^{\cal y}$ we define the following associated operator
\begin{equation}
\pi_{\Lambda} := \sum_{k} A_k \otimes \overline{A}_k.
\end{equation}
The eigenvalues and eigenvectors of $\Lambda$ are given by the pairs $(\lambda_X, X)$ such that $\Lambda(X) = \lambda_{X} X$. The operator $\pi_{\Lambda}$ is useful because it has the same eigenvalues as $\Lambda$ and its eigenvectors are in one-to-one correspondence with the eigenvectors of $\Lambda$: If $X$ is an eigenvector of $\Lambda$, then $(\id \otimes X)\ket{\Phi}$ is an eigenvector of $\pi_{\Lambda}$, with $\ket{\Phi}$ the maximally entangled state.

The following lemma was first given in Ref. \cite{FNW92} and we present here a proof in order to flesh out the dependency on the bond dimension of the state.
\begin{lemma}
Let $\Lambda : {\cal D}(\mathbb{C}^D) \rightarrow  {\cal D}(\mathbb{C}^D)$ be a unital channel with Kraus decomposition $\Lambda(\rho) = \sum_{k}^d A_k \rho A_k^{\cal y}$ and second largest eigenvalue $\eta$. Consider the matrix product state:
\begin{equation}
\ket{\psi}_{1, ..., n} := \sum_{i_1=1}^d...\sum_{i_n=1}^d \tr(A_{i_1}...A_{i_n}) \ket{i_1, ..., i_n}.
\end{equation}
Then if $A = [1, r]$, $B = [r+1, r+l+1]$, and $C = [r+l+2, n]$,
\begin{equation}
\text{Cor} (A:C) \leq \Vert \rho_{AC} - \rho_A \otimes \rho_C \Vert_1 \leq D \eta^{l}.
\end{equation}
\end{lemma}
\begin{proof}
Let $V : {\cal H}_1 \rightarrow {\cal H}_1 \otimes {\cal H}_2$ be an isometric extension of $\Lambda$, with ${\cal H}_1 \cong \mathbb{C}^D$ and ${\cal H}_2 \cong \mathbb{C}^d$, i.e. $\Lambda(\rho) = \tr_{ {\cal H}_2 }\left( V \rho V^{\cal y} \right)$. The reduced density matrix $\rho_{AC}$ of $\ket{\psi}_{ABC}$ is given by
\begin{equation}
\rho_{AC} =  \tr_{{\cal H}_1} \left(V^{n_A} \left(\Lambda^{l} \left(V^{n_C}(\tau) \right) \right)\right),
\end{equation}
with $n_A$ and $n_C$ the number of sites in $A$ and $C$, respectively, and $\tau$ the maximally mixed state in ${\cal D}(\mathbb{C}^D)$.  We used the notation
\begin{equation}
V^{k}(\tau) := \left(V \circ ... \circ V\right) \rho \left(V^{\cal y} \circ ... \circ V^{\cal y} \right),
\end{equation}
with the $k$ fold composition of $V$, and likewise for $\Lambda^{l}$. 

Let $\Phi : {\cal D}(\mathbb{C}^D) \rightarrow  {\cal D}(\mathbb{C}^D)$ be the completely depolarized channel: $\Phi (\rho) = \tau$. Then
\begin{equation}
\rho_{A} \otimes \rho_C =  \tr_{{\cal H}_1} \left(V^{n_A} \left(\Phi^{l} \left(V^{n_C}(\tau) \right) \right)\right).
\end{equation}
Therefore
\be
\Vert \rho_{AC} - \rho_A \otimes \rho_C \Vert_1 \leq \Vert  \Lambda^{l} \left(V^{n_C}(\tau) \right)  -   \Phi^{l} \left(V^{n_C}(\tau) \right)  \Vert_1 \leq \Vert \Lambda^{l} - \Phi \Vert_{\Diamond},
\ee
where the first inequality follows from the monotonicity of trace norm under trace preserving CP maps and the second inequality from the definition of the diamond norm. Using the relation
\begin{equation}
\Vert \Lambda_1 - \Lambda_2 \Vert_{\Diamond} \leq D \Vert \Lambda_1 - \Lambda_2 \Vert_{2 \rightarrow 2} = D\Vert \pi_{\Lambda_1} - \pi_{\Lambda_2} \Vert_{\infty},
\end{equation}
valid for every two channels $\Lambda_1, \Lambda_2 : {\cal D}(\mathbb{C}^D) \rightarrow {\cal D}(\mathbb{C}^D)$, we get
\begin{equation}
 \Vert \Lambda^{l} - \Phi \Vert_{\Diamond} \leq  D \Vert \pi_{\Lambda^{l}} - \pi_{\Phi} \Vert_{\infty} =  D \left \Vert \left(\pi_{\Lambda}\right)^{l} - \pi_{\Phi} \right \Vert_{\infty} = D \eta^{l},
\end{equation}
where we used that $\pi_{\Phi}$ is the projector onto the maximum eigenvector of $\pi_{\Lambda}$ (of value one and corresponding eigenvector $\tau$).
\end{proof}

\section{Correlations in Quantum Expander States} \label{CorrelationsExpander}

In this section we show how the results of Ref. \cite{Has07b} imply there are correlations between regions separated by of order $\log(D)/\log(d)$ sites in quantum expander states with matrices given by independent Haar unitaries. We start with the following lemma,

\begin{lemma} \label{aubexpanders}
Let $\Lambda : {\cal D}(\mathbb{C}^D) \rightarrow  {\cal D}(\mathbb{C}^D)$ be a unital channel with Kraus decomposition $\Lambda(\rho) = \sum_{k=1}^d A_k \rho A_k^{\cal y}$. Consider the matrix product state 
\begin{equation} \label{expanderstateappendix}
\ket{\psi}_{1, ..., n} := \sum_{i_1=1}^d...\sum_{i_n=1}^d \tr(A_{i_1}...A_{i_n}) \ket{i_1, ..., i_n}.
\end{equation}
Then the reduced density matrix $\rho_l$ of $l$ sites is such that
\be
\label{eq:tr_rho_E1}
\tr \left( \rho^2_{l} \right)= \frac{1}{D^2} \sum_{i,j}\tr \left( \Lambda^l(|i\>\<j|)\Lambda^l(|j\>\<i|) \right).
\ee
\end{lemma}
\begin{proof}
We have 
\be
\rho_l=\sum^d_{{i_1\ldots i_l=1\atop j_1\ldots j_l=1}}
\tr \left(A_{i_1}\ldots A_{i_l} {\id\over D} A_{j_1}^\dagger\ldots A_{j_l}^\dagger \right)
|i_1\>\<j_1|\ot \ldots \ot |i_l\>\<j_l|.
\ee

Note that this state is reduced density matrix $\rho_E$ of a tripartite state $\ket{\psi}_{ABE}$, with $|A|=|B|=D, |E|=d^l$, defined as
\be
\ket{\psi}_{ABE} := \frac{1}{d^{l/2}}\sum^d_{i_1, \ldots, i_l=1}(\id_A \ot A_{i_1}\ldots A_{i_l})|\Phi\>_{AB}|i_1,\ldots, i_l\>_E,
\ee
where the product of operators $A_{i_j}$ acts on system $B$ and 
\be
\ket{\Phi}_{AB}=\frac{1}{\sqrt{D}}\sum_{k=1}^D|k\>_A|k\>_B
\ee
is the maximally entangled state on $AB$. One then finds that the $AB$ reduced density matrix of this state reads
\be \label{Formula6}
\rho_{AB}=(\id_A \ot \Lambda^l_B)(|\Phi\>\<\Phi|).
\ee
Eq. \eqref{eq:tr_rho_E1} then follows from $\tr\left( \rho_l^2 \right) = \tr\left( \rho_E^2 \right) = \tr\left( \rho_{AB}^2 \right)$ and Eq. (\ref{Formula6}).
\end{proof}

The next lemma follows from the previous lemma and the results of \cite{Has07b, Has12}.

\begin{lemma} \label{hasentropy}
Let $A_i := \frac{1}{\sqrt{d}} U_i$, with $U_i$ chosen independently at random according to the Haar measure. Then there exists constants $\gamma > 0, k > 1$ such that the state given by Eq. (\ref{expanderstateappendix}) satisfies
\be
\<  \tr \left( \rho_l^2 \right)  \> \leq \frac{1}{D^2} + \frac{kl}{d^l}.
\ee
for any $l$ such that $l \leq D^{\gamma}$, where $\<\ldots \>$ denotes the average over $U_i$'s. 
\end{lemma}

\begin{proof} 
For $A_i =  \frac{1}{\sqrt{d}} U_i$ Hastings proved an upper bound for the average of the RHS of Eq. \eqref{eq:tr_rho_E1} \cite{Has07b, Has12}. The bound says there are constants $\gamma > 0, k > 1$ such that  
\be
\sum_{i,j}\tr \left( \Lambda^l(|i\>\<j|)\Lambda^l(|j\>\<i|) \right) \leq 1 + kl\left( \frac{1}{\sqrt{d}}   \right)^{2l} D^2
\ee
for any $l \leq D^{\gamma}$. Using Lemma \ref{aubexpanders} we get $\tr \left(\rho_{l}^2 \right)\leq \frac{1}{D^2} + \frac{kl}{d^l}$.
\end{proof}

We are now ready to prove the main result of this section. The idea is to reduce the problem to one similar to the case of random states, where one could argue there were correlations by finding two neighbouring regions $AB$ which were decoupled. Here there are no neighbouring regions which are decoupled a priori. However using the result of Lemma \ref{hasentropy} that regions of size $O(\log(D)/\log(d))$ have very large entropy, we will be able to show that one can find a projector that has large probability of being measured in $A$, and that decouples it from its neighbouring region $B$. This idea, of showing the existence of correlations by decoupling one subsystem from its neighbouring one by measurements, will be the central idea in the proof of Theorem \ref{maintheorempure}. Here we consider the particular case in which the entropy is almost the maximum value possible, what makes the argument much simpler. 
\begin{prop}
Let $A_i := \frac{1}{\sqrt{d}} U_i$, with $U_i$ chosen independently at random according to the Haar measure. Then the state given by Eq. (\ref{expanderstateappendix}) has $\text{Cor}(A:C) \geq \Omega \left(1/l \right)$ between regions $A$ and $C$ separated by $l$ sites for any $l \leq O(\log(D)/\log(d))$.
\end{prop}

\begin{proof}
Consider a region $\{1, ..., l\}$ with $l = O(\log(D)/\log(d))$ and divide it into subregions $A$ and $B$ of equal sizes, as in Fig. \ref{fig0}. Let $\varepsilon := 0.01$. Applying Lemma \ref{decouplingforpoors} with $N = kl/\varepsilon$, and using Lemma \ref{purifiedvstracenorm}, it follows there is a projector $Q$ acting on subsystem $A$, of dimension bigger than two, such that $\bra{\psi} (Q \otimes \id_{BC}) \ket{\psi} \geq O(l/\varepsilon)$ and 
\begin{equation}
D\left ( \rho_{A'B}, \tau_{A'} \otimes \rho_B \right ) \leq \sqrt{2\delta},
\end{equation}
with $\rho_{A'B}$ the reduced density matrix of the postselected state $\ket{\phi}_{A'BC} := (Q \otimes \id_{BC}) \ket{\psi}_{ABC} / \sqrt{\bra{\psi} (Q \otimes \id_{BC}) \ket{\psi}}$, and 
\be
\delta := 2 \sqrt{ \frac{d^l \varepsilon}{D^2 k l} + \varepsilon } + \frac{2 \varepsilon}{kl}. 
\ee
Choosing $l = \log(D)/\log(d)$  we get $\delta \leq 3\sqrt{2\varepsilon}$.

From Uhlmann's theorem (Lemma \ref{purifieddistancepurification} in Appendix \ref{purifieddistance}) we find there is an isometry $V : C \rightarrow C_1C_2$ that can be applied to $C$ such that
\begin{equation} \label{uhlmandecoupling}
D( (\id_{AB} \otimes V) \ket{\psi}_{ABC}\bra{\psi}(\id_{AB} \otimes V)^{\cal y}, \ket{\Phi}_{AC_1}\bra{\Phi} \otimes \ket{\Psi}_{BC_2}\bra{\Phi}) \leq \sqrt{2 \varepsilon},
\end{equation}
with $\ket{\Phi}_{AC_1} = \text{dim(A)}^{-1/2} \sum_{k=1}^{\text{dim}(A)} \ket{k, k}$ a maximally entangled state between $A$ and $C_1$ and $\ket{\Psi}_{BC_2}$ a purification of $\rho_B$. Thus defining $M =  \sum_{i=1}^{\text{dim}(Q)/2} \ket{k}\bra{k}$ and $N = V \left( M \otimes \id_{C_2} \right) V^{\cal y}$, we find 
\begin{equation}
\text{Cor}(A:C) \geq \tr \left((QM_AQ \otimes N_C)(\rho_{AC} - \rho_A \otimes \rho_C)\right) \geq \frac{\varepsilon}{8l}. 
\end{equation}
Therefore there are strong correlations between the regions $A$ and $C$ separated by $l/2 =  \log(D)/(2\log(d))$ sites.
\end{proof}

\begin{lemma}[Proposition 4 of \cite{HOW07}] \label{decouplingforpoors}
Let $\ket{\psi}_{ABC}$ be a pure state. Then there exists a POVM $\{ P_k \}_{k=1}^{N+1}$ consisting of $N = \left \lfloor \frac{|A|}{L}  \right\rfloor$ projectors of rank $L$ and one of rank $|A| - NL \leq L$ such that
\be
\sum_k p_k \left \Vert \rho^k_{A'B} - \tau_{A'} \otimes \rho_B \right \Vert_1 \leq 2 \sqrt{L |B| \tr \left(\rho_{AB}^2 \right)} + 2 \frac{L}{|A|},
\ee
with $p_k := \bra{\psi} P_{k} \otimes \id_{BC} \ket{\psi}$ and $\rho_{A'B}^k$ the reduced density matrix of the postselected state $\ket{\phi}_{A'BC} := \left( P_k \otimes \id_{BC}   \right) \ket{\psi}_{ABC} / \sqrt{p_k}$.
\end{lemma}

\section{Single-shot State Entanglement Distillation} \label{singleshotmergingsection}

We start this appendix with a formal definition entanglement distillation. In fact we will consider the more demanding task of single-shot state merging, which not only distills EPR pairs but also transfers the state of one of the parties to the other one. Although we do not make use of this further property in the proof, we use it here in order to state the results of Ref. \cite{DBWR10} that we make use in the form they originally appeared.  
\begin{definition} \label{statemergingdef}
Consider a tripartite state $\ket{\psi}_{ABC}$. Let $A_1$ and $C_1$ be registers each of dimension $L$. A measurement defined by the POVM elements $\{ M_k \}_{k=1}^N$, with 
$M_k : A \rightarrow A_1$, and a set of isometries $\{ V_k \}_{k=1}^N$, with $V_k : C \rightarrow C_1 C' C$, define a $(L, N, \varepsilon)$-quantum state merging protocol for $\ket{\psi}_{ABC}$ with error $\varepsilon$, entanglement distillation rate $\log(L)$, and classical communication cost $\log(N)$, if 
\begin{equation}
D\left( \sum_{k=1}^N (\sqrt{M_k} \otimes \id_B \otimes V_k)\ket{\psi}_{ABC}\bra{\psi} (\sqrt{M_k} \otimes \id_B \otimes V_k )^{\cal y}, \ket{\Phi}_{A_1C_1}\bra{\Phi} \otimes \ket{\psi}_{BC'C} \bra{\psi} \right) \leq \varepsilon
\end{equation} 
\end{definition}
We note that in its most general form, the state merging protocol also makes sense when one must use pre-shared entanglement in order to transfer the $A$ part of the state to $C$. The application of the protocol in this work, however, does not concern this regime. Note also that as we already mentioned, for our application we do not need the state merging part of the protocol per se, but only the entanglement distillation part. However, it is very important that we distill entanglement with the specified classical communication cost achieved in the single-shot version \cite{DBWR10} of the state merging protocol of \cite{HOW05, HOW07}, i.e. 

\begin{lemma}[Single-Shot State Merging \cite{DBWR10}] \label{singleshotmerging}
Given a tripartite state $\ket{\psi}_{ABC}$, there is protocol for quantum state merging with 
\begin{equation}
\log(N) \leq H_{\max}^{\varepsilon}(A) - H_{\min}^{\varepsilon}(A|B) - 4 \log \varepsilon + 2 \log 13,
\end{equation}
\begin{equation}
\log(L) \leq H_{\min}^{\varepsilon}(A|B) - 4 \log \varepsilon + 2 \log 13.
\end{equation}
and error $\delta  = 13 \sqrt{\varepsilon}$.
\end{lemma}

\end{document}